\documentclass[11pt]{llncs}
\usepackage{amsmath,amssymb,times,a4wide}
\usepackage[final]{graphicx}

\def\floor#1{\mathop{\lt\lfloor#1\rt\rfloor}}

\def\poly{\mathrm{poly}}

\def\rt{\right}
\def\lt{\!\left}

\def\union{\cup}

\def\NP{\mathrm{NP}}
\def\PoA{\mathrm{PoA}}

\def\Th{\mathrm{\Theta}}
\def\Om{\mathrm{\Omega}}

\def\eps{\varepsilon}
\def\e{\epsilon}

\def\reals{\mathrm{I\!R}}

\def\P{\mathcal{P}}
\def\G{\mathcal{G}}
\def\H{\mathcal{H}}
\def\I{\mathcal{I}}

\def\Bst{B^\ast}
\def\ft{\tilde{f}}
\def\gt{\tilde{g}}
\def\Bt{\tilde{B}}

\def\Pone{\mathrm{ParRidBC}}
\def\Ptwo{\mathrm{BSubNBC}}
\def\DD{\mathrm{2\mbox{-}DDP}}
\def\YES{\mbox{\sc yes}}
\def\NO{\mbox{\sc no}}

\def\fig#1#2#3{\begin{figure}[t]%
\centerline{\includegraphics[width=#1]{#2}}%
\vspace*{-2mm}\caption{#3}\vspace*{-5mm}\end{figure}}

\title{On the Hardness of Network Design for Bottleneck Routing Games%
\thanks{This work was supported by the project Algorithmic Game Theory, co-financed by the European Union (European Social Fund - ESF) and Greek national funds, through the Operational Program ``Education and Lifelong Learning'', under the research funding program Thales, by an NTUA Basic Research Grant (PEBE 2009), by the ERC project RIMACO, and by the EU-FP7 Project e-Compass.}}

\author{Dimitris Fotakis\inst{1} \and Alexis C. Kaporis\inst{2}
\and Thanasis Lianeas\inst{1} \and Paul G. Spirakis\inst{3,4}}

\institute{School of Electrical and Computer Engineering, National Technical
University of Athens, 15780 Athens, Greece.
\and
Department of Information and Communication Systems Engineering,
University of the Aegean, Greece.
\and
Department of Computer Engineering and Informatics, University of Patras,
26500 Patras, Greece.
\and
Research Academic Computer Technology Institute, N.~Kazantzaki Str.,
University Campus, 26500 Patras, Greece.\\
Email:~{\tt fotakis@cs.ntua.gr, kaporisa@gmail.com, tlianeas@mail.ntua.gr, spirakis@cti.gr}}

\begin{document}
\maketitle

\begin{abstract}
In routing games, the selfish behavior of the players may lead to a degradation of the network performance at equilibrium. In more than a few cases however, the equilibrium performance can be significantly improved if we remove some edges from the network. This counterintuitive fact, widely known as Braess's paradox, gives rise to the (selfish) network design problem, where we seek to recognize routing games suffering from the paradox, and to improve their equilibrium performance by edge removal.
In this work, we investigate the computational complexity and the approximability of the network design problem for non-atomic bottleneck routing games, where the individual cost of each player is the bottleneck cost of her path, and the social cost is the bottleneck cost of the network, i.e. the maximum latency of a used edge.
We first show that bottleneck routing games do not suffer from Braess's paradox either if the network is series-parallel, or if we consider only subpath-optimal Nash flows.
On the negative side, we prove that even for games with strictly increasing linear latencies, it is $\NP$-hard not only to recognize instances suffering from the paradox, but also to distinguish between instances for which the Price of Anarchy ($\PoA$) can decrease to $1$ and instances for which the $\PoA$ cannot be improved by edge removal, even if their $\PoA$ is as large as $\Om(n^{0.121})$. This implies that the network design problem for linear bottleneck routing games is $\NP$-hard to approximate within a factor of $O(n^{0.121-\eps})$, for any constant $\eps > 0$. The proof is based on a recursive construction of hard instances that carefully exploits the properties of bottleneck routing games, and may be of independent interest.
On the positive side, we present an algorithm for finding a subnetwork that is almost optimal w.r.t. the bottleneck cost of its worst Nash flow, when the worst Nash flow in the best subnetwork routes a non-negligible amount of flow on all used edges. We show that the running time is essentially determined by the total number of paths in the network, and is quasipolynomial when the number of paths is quasipolynomial.
\end{abstract}

\pagestyle{plain}
\pagenumbering{arabic}
\thispagestyle{empty}
\setcounter{page}{0}
\newpage

\section{Introduction}
\label{s:intro}

An typical instance of a non-atomic \emph{bottleneck routing game} consists of a directed network, with an origin $s$ and a destination $t$, where each edge is associated with a non-decreasing function that determines the edge's latency as a function of its traffic. A rate of traffic is controlled by an infinite population of players, each willing to route a negligible amount of traffic through an $s-t$ path. The players are non-cooperative and selfish, and seek to minimize the maximum edge latency, a.k.a. the \emph{bottleneck cost} of their path. Thus, the players reach a \emph{Nash equilibrium flow}, or simply a \emph{Nash flow}, where they all use paths with a common locally minimum bottleneck cost. Bottleneck routing games and their variants have received considerable attention due to their practical applications in communication networks (see e.g., \cite{CGK05,BO07} and the references therein).

\smallskip{\bf Previous Work and Motivation.}
Every bottleneck routing game is known to admit a Nash flow that is optimal for the network, in the sense that it minimizes the maximum latency on any used edge, a.k.a. the bottleneck cost of the network (see e.g., \cite[Corollary~2]{BO07}). On the other hand, bottleneck routing games usually admit many different Nash flows, some with a bottleneck cost quite far from the optimum. Hence, there has been a considerable interest in quantifying the performance degradation due to the players' non-cooperative and selfish behavior in (several variants of) bottleneck routing games. This is typically measured by the \emph{Price of Anarchy} ($\PoA$) \cite{KP99}, which is the ratio of the bottleneck cost of the worst Nash flow to the optimal bottleneck cost of the network.

Simple examples (see e.g., \cite[Figure~2]{CDR06}) demonstrate that the $\PoA$ of bottleneck routing games with linear latency functions can be as large as $\Om(n)$, where $n$ is the number of vertices of the network.
For atomic splittable bottleneck routing games, where the population of players is finite, and each player controls a non-negligible amount of traffic which can be split among different paths, Banner and Orda \cite{BO07} observed that the $\PoA$ can be unbounded, even for very simple networks, if the players have different origins and destinations and the latency functions are exponential. On the other hand, Banner and Orda proved that if the players use paths that, as a secondary objective, minimize the number of bottleneck edges, then all Nash flows are optimal.
For a variant of non-atomic bottleneck routing games, where the social cost is the average (instead of the maximum) bottleneck cost of the players, Cole, Dodis, and Roughgarden \cite{CDR06} proved that the $\PoA$ is $4/3$, if the latency functions are affine and a subclass of Nash flows, called \emph{subpath-optimal Nash flows}, is only considered. Subsequently, Mazalov et al. \cite{MMST06} studied the inefficiency of the best Nash flow under this notion of social cost.

For atomic unsplittable bottleneck routing games, where each player routes a unit of traffic through a single $s-t$ path, Banner and Orda \cite{BO07} proved that for polynomial latency functions of degree $d$, the $\PoA$ is $O(m^d)$, where $m$ is the number of edges of the network. On the other hand, Epstein, Feldman, and Mansour \cite{EFM09} proved that for series-parallel networks with arbitrary latency functions, all Nash flows are optimal.
Subsequently, Busch and Magdon-Ismail \cite{BMI09} proved that the $\PoA$ of atomic unsplittable bottleneck routing games with identity latency functions can be bounded in terms of natural topological properties of the network. In particular, they proved that the $\PoA$ of such games is bounded from above by $O(l+\log n)$, where $l$ is the length of the longest $s-t$ path, and by $O(k^2+\log^2 n)$, where $k$ is length of the longest circuit. 

With the $\PoA$ of bottleneck routing games so high and crucially depending on topological properties of the network, a natural approach to improving the performance at equilibrium is to exploit the essence of Braess's paradox \cite{Bra68}, namely that removing some edges may change the network topology (e.g., it may decrease the length of the longest path or cycle), and significantly improve the bottleneck cost of the worst Nash flow (see e.g., Fig.~\ref{fig:braess}). This approach gives rise to the (selfish) \emph{network design problem}, where we seek to recognize bottleneck routing games suffering from the paradox, and to improve the bottleneck cost of the worst Nash flow by edge removal. In particular, given a bottleneck routing game, we seek for the \emph{best subnetwork}, namely, the subnetwork for which the bottleneck cost of the worst Nash flow is best possible.
In this setting, one may distinguish two extreme classes of instances: \emph{paradox-free} instances, where edge removal cannot improve the bottleneck cost of the worst Nash flow, and \emph{paradox-ridden} instances, where the bottleneck cost of the worst Nash flow in the best subnetwork is equal to the optimal bottleneck cost of the original network (see also \cite{Rou06,FKS09}).

The approximability of selective network design, a generalization of network design where we cannot remove certain edges, was considered by Hou and Zhang \cite{HZ07}. For atomic unsplittable bottleneck routing games with a different traffic rate and a different origin and destination for each player, they proved that if the latency functions are polynomials of degree $d$, it is $\NP$-hard to approximate selective network design within a factor of $O(m^{d-\eps})$, for any constant $\eps > 0$. Moreover, for atomic $k$-splittable bottleneck routing games with multiple origin-destination pairs, they proved that selective network design is $\NP$-hard to approximate within any constant factor.

However, a careful look at the reduction of \cite{HZ07} reveals that their strong inapproximability results crucially depend on both (i) that we can only remove certain edges from the network, so that the subnetwork actually causing a high $\PoA$ cannot be destroyed, and (ii) that the players have different  origins and destinations (and also are atomic and have different traffic rates).
As for the importance of (ii), in a different setting, where the players' individual cost is the sum of edge latencies on their path and the social cost is the bottleneck cost of the network, it is known that Braess's paradox can be dramatically more severe for instances with multiple origin-destination pairs than for instances with a single origin-destination pair. More precisely, Lin et al. \cite{LRT04} proved that if the players have a common origin and destination, the removal of at most $k$ edges from the network cannot improve the equilibrium bottleneck cost by a factor greater than $k+1$. On the other hand, Lin et al. \cite{LRTW05} presented an instance with two origin-destination pairs where the removal of a single edge improves the the equilibrium bottleneck cost by a factor of $2^{\Om(n)}$.
Therefore, both at the technical and at the conceptual level, the inapproximability results of \cite{HZ07} do not really shed light on the approximability of the (simple, non-selective) network design problem in the simplest, and most interesting, setting of non-atomic bottleneck routing games with a common origin-destination pair for all players.

\fig{0.7\textwidth}{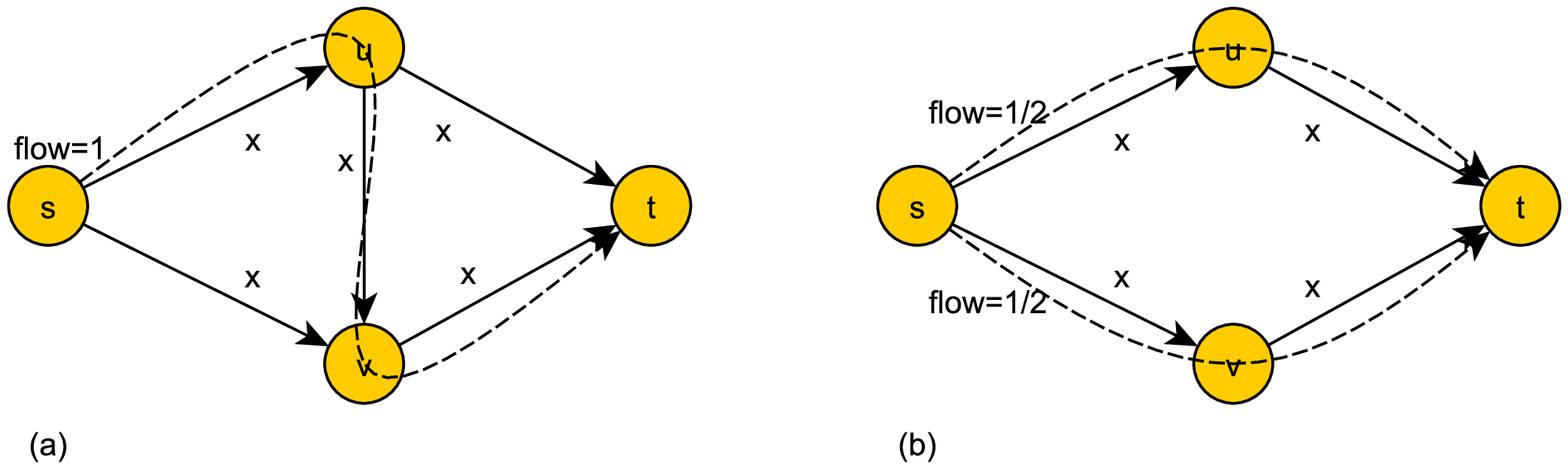}%
{\label{fig:braess}An example of Braess's paradox for bottleneck routing games. We consider a routing instance with identity latency functions and a unit of traffic to be routed from $s$ to $t$. The worst Nash flow, in (a), routes all flow through the path $(s, u, v, t)$, and has a bottleneck cost of $1$. On the other hand, the optimal flow routes $1/2$ unit through the path $(s, u, t)$ and $1/2$ unit through the path $(s, v, t)$, and has a bottleneck cost of $1/2$. Hence, $\PoA = 2$. In the subnetwork (b), obtained by removing the edge $(u, v)$, we have a unique Nash flow that coincides with the optimal flow, and thus the $\PoA$ becomes $1$. Hence the network on the left is \emph{paradox-ridden}, and the network on the right is the \emph{best subnetwork} of it.}

\smallskip\noindent{\bf Contribution.}
Hence, in this work, we investigate the approximability of the network design problem for the simplest, and seemingly easier to approximate, variant of non-atomic bottleneck routing games (with a single origin-destination pair). Our main result is that network design is hard to approximate within reasonable factors, and holds even for the special case of strictly increasing linear latencies. To the best of our knowledge, this is the first work that investigates the impact of Braess's paradox and the approximability of the network design problem for the basic variant of bottleneck routing games.

In Section~\ref{s:easy-cases}, we use techniques similar to those in \cite{EFM09,CDR06}, and show that bottleneck routing games do not suffer from Braess's paradox either if the network is series-parallel, or if we consider only subpath-optimal Nash flows.

On the negative side, we employ, in Section~\ref{s:np-hard}, a reduction from the $2$-Directed Disjoint Paths problem, and show that for linear bottleneck routing games, it is $\NP$-hard to recognize paradox-ridden instances (Lemma~\ref{l:gap}). In fact, the reduction shows that it is $\NP$-hard to distinguish between paradox-ridden instances and paradox-free instances, even if their $\PoA$ is equal to $4/3$, and thus, it is $\NP$-hard to approximate the network design problem within a factor less than $4/3$.

In Section~\ref{s:inapprox}, we apply essentially the same reduction, but in a recursive way, and obtain a much stronger inapproximability result. In particular, we assume the existence of a $\gamma$-gap instance, which establishes that network design is inapproximable within a factor less than $\gamma$, and show that the construction of Lemma~\ref{l:gap}, but with some edges replaced by copies of the gap instance, amplifies the inapproximability threshold by a factor of $4/3$, while it increases the size of the network by roughly a factor of $8$ (Lemma~\ref{l:rec-gap}). Therefore, starting from the $4/3$-gap instance of Lemma~\ref{l:gap}, and recursively applying this construction a logarithmic number times, we show that it is $\NP$-hard to approximate the network design problem for linear bottleneck routing games within a factor of $O(n^{0.121-\eps})$, for any constant $\eps > 0$. An interesting technical point is that we manage to show this inapproximability result, even though we do not know how to efficiently compute the worst equilibrium bottleneck cost of a given subnetwork. Hence, our reduction uses a certain subnetwork structure to identify good approximations to the best subnetwork. To the best of our knowledge, this is the first rime that a similar recursive construction is used to amplify the inapproximability threshold of the network design problem, and of any other optimization problem related to selfish routing.

In Section~\ref{s:sparsification}, we consider latency functions that satisfy a Lipschitz condition, and present an algorithm for finding a subnetwork that is almost optimal w.r.t. the bottleneck cost of its worst Nash flow, when the worst Nash flow in the best subnetwork routes a non-negligible amount of flow on all used edges. The algorithm is based on Alth\"ofer's Sparcification Lemma \cite{Alt94}, and is motivated by its recent application to network design for additive routing games \cite{FKS09}. For any constant $\eps > 0$, the algorithm computes a subnetwork and an $\eps/2$-Nash flow whose bottleneck cost is within an additive term of $O(\eps)$ from the worst equilibrium bottleneck cost in the best subnetwork. The running time is roughly $|\P|^{\poly(\log m)/\eps^2}$, 
%
and is quasipolynomial, when the number $|\P|$ of paths is quasipolynomial.

\smallskip\noindent{\bf Other Related Work.}
Considerable attention has been paid to the approximability of the network design problem for \emph{additive routing games}, where the players seek to minimize the sum of edge latencies on their path, and the social cost is the total latency incurred by the players.
In fact, Roughgarden \cite{Rou06} first introduced the selfish network design problem in this setting, and proved that it is $\NP$-hard to recognize paradox-ridden instances. Roughgarden also proved that it is $\NP$-hard to approximate the network design problem for additive routing games within a factor less than $4/3$ for affine latencies, and less than $\floor{n/2}$ for general latencies.
For atomic unsplittable additive routing games with weighted players, Azar and Epstein \cite{AzaEps05} proved that network design is $\NP$-hard to approximate within a factor less than $2.618$, for affine latencies, and less than $d^{\Th(d)}$, for polynomial latencies of degree $d$.

On the positive side, Milchtaich \cite{Mil06} proved that non-atomic additive routing games on series-parallel networks do not suffer from Braess's paradox.
Fotakis, Kaporis, and Spirakis \cite{FKS09} proved that we can efficiently recognize paradox-ridden instances when the latency functions are affine, and all, but possibly a constant number of them, are strictly increasing. Moreover, applying Alth\"ofer's Sparsification Lemma \cite{Alt94}, they gave an algorithm that approximates network design for affine additive routing games within an additive term of $\eps$, for any constant $\eps > 0$, in time that is subexponential if the total number of $s-t$ paths is polynomial and all paths are of polylogarithmic length.

\section{Model, Definitions, and Preliminaries}
\label{s:model}

\noindent{\bf Routing Instances.}
A \emph{routing instance} is a tuple $\G = (G(V, E), (c_e)_{e \in E}, r)$, where $G(V, E)$ is a directed network with an origin $s$ and a destination $t$, $c_e : [0, r] \mapsto \reals_{\geq 0}$ is a continuous non-decreasing latency function associated with each edge $e$, and $r > 0$ is the traffic rate entering at $s$ and leaving at $t$. We let $n \equiv |V|$ and $m \equiv |E|$, and let $\P$ 
denote the set of simple $s-t$ paths in $G$.
%
%
A latency function $c_e(x)$ is \emph{linear} if $c_e(x) = a_e x$, for some $a_e > 0$, and \emph{affine} if $c_e(x) = a_e x + b_e$, for some $a_e, b_e \geq 0$.
%
%
We say that a latency function $c_e(x)$ satisfies the \emph{Lipschitz condition} with constant $\xi > 0$, if for all $x, y \in [0, r]$, $|c_e(x) - c_e(y)| \leq \xi |x - y|$.

\smallskip\noindent{\bf Subnetworks and Subinstances.}
Given a routing instance $\G = (G(V, E), (c_e)_{e \in E}, r)$, any subgraph $H(V, E')$, $E' \subseteq E$, obtained from $G$ by edge deletions, is a \emph{subnetwork} of $G$. $H$ has the same origin $s$ and destination $t$ as $G$, and the edges of $H$ have the same latency functions as in $\G$. Each instance $\H = (H(V, E'), (c_e)_{e \in E'}, r)$, where $H(V, E')$ is a subnetwork of $G(V, E)$, is a \emph{subinstance} of $\G$.

\smallskip\noindent{\bf Flows.}
A ($\G$-feasible) \emph{flow} $f$ is a non-negative vector indexed by $\P$ so that $\sum_{p \in \P} f_p = r$. For a flow $f$ and each edge $e$, we let $f_e = \sum_{p: e \in p} f_p$ denote the amount of flow that $f$ routes through $e$.
A path $p$ (resp. edge $e$) is used by flow $f$ if $f_p > 0$ (resp. $f_e > 0$). Given a flow $f$, the latency of each edge $e$ is $c_e(f_e)$, and the \emph{bottleneck cost} of each path $p$ is $b_p(f) = \max_{e \in p} c_e(f_e)$.
The \emph{bottleneck cost} of a flow $f$, denoted $B(f)$, is
$B(f) = \max_{p: f_p > 0} b_p(f)$, i.e., the maximum bottleneck cost of any used path.

\smallskip\noindent{\bf Optimal Flow.}
An \emph{optimal} flow of an instance $\G$, denoted $o$, minimizes the bottleneck cost among all $\G$-feasible flows. We let $\Bst(\G) = B(o)$. We note that for every subinstance $\H$ of $\G$, $\Bst(\H) \geq \Bst(\G)$.

\smallskip\noindent{\bf Nash Flows and their Properties.}
We consider a non-atomic model of selfish routing, where the traffic is divided among an infinite population of players, each routing a negligible amount of traffic from $s$ to $t$. A flow $f$ is at \emph{Nash equilibrium}, or simply, is a \emph{Nash flow}, if $f$ routes all traffic on paths of a locally minimum bottleneck cost. Formally, $f$ is a Nash flow if for all $s-t$ paths $p, p'$, if $f_p > 0$, then $b_p(f) \leq b_{p'}(f)$. Therefore, in a Nash flow $f$, all players incur a common bottleneck cost $B(f) = \min_{p} b_p(f)$, and for every $s-t$ path $p'$, $B(f) \leq b_p'(f)$.

We observe that if a flow $f$ is a Nash flow for an $s-t$ network $G(V, E)$, then the set of edges $e$ with $c_e(f_e) \geq B(f)$ comprises an $s-t$ cut in $G$. For the converse, if for some flow $f$, there is an $s-t$ cut consisting of edges $e$ either with $f_e > 0$ and $c_e(f_e) = B(f)$, or with $f_e = 0$ and $c_e(f_e) \geq B(f)$, then $f$ is a Nash flow. Moreover, for all bottleneck routing games with linear latencies $a_e x$, a flow $f$ is a Nash flow iff the set of edges $e$ with $c_e(f_e) = B(f)$ comprises an $s-t$ cut.
%

It can be shown that every bottleneck routing game admits at least one Nash flow (see e.g., \cite[Proposition~2]{CDR06}), and that there is an optimal flow that is also a Nash flow (see e.g., \cite[Corollary~2]{BO07}). In general, a bottleneck routing game admits many different Nash flows, each with a possibly different bottleneck cost of the players.
Given an instance $\G$, we let $B(\G)$ denote the bottleneck cost of the players in the worst Nash flow of $\G$, i.e. the Nash flow $f$ that maximizes $B(f)$ among all Nash flows. We refer to $B(\G)$ as the worst equilibrium bottleneck cost of $\G$. For convenience, for an instance $\G = (G, c, r)$, we sometimes write $B(G, r)$, instead of $B(\G)$, to denote the worst equilibrium bottleneck cost of $\G$.
We note that for every subinstance $\H$ of $\G$, $\Bst(\G) \leq B(\H)$, and that there may be subinstances $\H$ with $B(\H) < B(\G)$, which is the essence of Braess's paradox (see e.g., Fig.~\ref{fig:braess}).

The following proposition considers the effect of a uniform scaling of the latency functions. For completeness, we include the proof in the Appendix, Section~\ref{app:s:scaling}.

\begin{proposition}\label{pr:scaling}
Let $\G = (G, c, r)$ be a routing instance, let $\alpha > 0$, and let $\G' = (G, \alpha c, r)$ be the routing instance obtained from $\G$ if we replace the latency function $c_e(x)$ of each edge $e$ with $\alpha c_e(x)$. Then, any $\G$-feasible flow $f$ is also $\G'$-feasible and has $B_{\G'}(f) = \alpha B_{\G}(f)$. Moreover, a flow $f$ is a Nash flow (resp. optimal flow) of $\G$ iff $f$ is a Nash flow (resp. optimal flow) of $\G'$.
\end{proposition}

\noindent{\bf Subpath-Optimal Nash Flows.}
For a flow $f$ and any vertex $u$, let $b_f(u)$ denote the minimum bottleneck cost of $f$ among all $s-u$ paths. The flow $f$ is a \emph{subpath-optimal Nash flow} \cite{CDR06} if for any vertex $u$ and any $s-t$ path $p$ with $f_p > 0$ that includes $u$, the bottleneck cost of the $s-u$ part of $p$ is $b_f(u)$.
For example, the Nash flow $f$ in Fig.~\ref{fig:braess}.a is not subpath-optimal, because $b_f(v) = 0$, through the edge $(s, v)$, while the bottleneck cost of the path $(s, u, v)$ is $1$. For this instance, the only subpath-optimal Nash flow is the optimal flow with $1/2$ unit on the path $(s, u, t)$ and $1/2$ unit on the path $(s, v, t)$.

\smallskip\noindent{\bf $\mathbf{\eps}$-Nash Flows.}
The definition of a Nash flow can be generalized to that of an ``almost Nash'' flow: For some constant $\eps > 0$, a flow $f$ is an $\eps$-Nash flow if for all $s-t$ paths $p$, $p'$, if $f_p > 0$, $b_p(f) \leq b_{p'}(f) + \eps$.

\smallskip\noindent{\bf Price of Anarchy.}
The \emph{Price of Anarchy} ($PoA$) of an instance $\G$, denoted $\rho(\G)$, is the ratio of the worst equilibrium bottleneck cost of $\G$ to the optimal bottleneck cost. Formally, $\rho(\G) = B(\G) / \Bst(\G)$.

\smallskip\noindent{\bf Paradox-Free and Paradox-Ridden Instances.}
A routing instance $\G$ is \emph{paradox-free} if for every subinstance $\H$ of $\G$, $B(\H) \geq B(\G)$. Paradox-free instances do not suffer from Braess's paradox and their $\PoA$ cannot be improved by edge removal.
If an instance is not paradox-free, edge removal can decrease the worst equilibrium bottleneck cost by a factor greater than $1$ and at most
$\rho(\G)$. An instance $\G$ is \emph{paradox-ridden} if there is a
subinstance $\H$ of $\G$ such that $B(\H) = \Bst(\G) = B(\G)/\rho(\G)$.
Namely, the $\PoA$ of paradox-ridden instances can decrease to $1$ by
edge removal.

\smallskip\noindent{\bf Best Subnetwork.}
Given an instance $\G = (G, c, r)$, the \emph{best subnetwork} $H^\ast$ of $G$ minimizes the worst equilibrium bottleneck cost, i.e., for all subnetworks $H$ of $G$, $B(H^\ast, r) \leq B(H, r)$.

\smallskip\noindent{\bf Problem Definitions.}
In this work, we investigate the complexity and the approximability of two fundamental selfish network design problems for bottleneck routing games:
\begin{itemize}
\item {\bf Paradox-Ridden Recognition} ($\Pone$)\,: Given an instance $\G$, decide if $\G$ is paradox-ridden.

\item {\bf Best Subnetwork} ($\Ptwo$)\,: Given an instance $\G$, find the best subnetwork $H^\ast$ of $G$.
\end{itemize}
We note that the objective function of $\Ptwo$ is the worst equilibrium bottleneck cost $B(H, r)$ of a subnetwork $H$. Thus, a (polynomial-time) algorithm $A$ achieves an $\alpha$-approximation for $\Ptwo$ if for all instances $\G$, $A$ returns a subnetwork $H$ with $B(H, r) \leq \alpha B(H^\ast, r)$.
A subtle point is that given a subnetwork $H$, we do not know how to efficiently compute the worst equilibrium bottleneck cost $B(H, r)$ (see also \cite{AzaEps05,HZ07}, where a similar issue arises). To deal with this delicate issue, our hardness results use a certain subnetwork structure to identify a good approximation to $\Ptwo$.


\smallskip\noindent{\bf Series-Parallel Networks.}
A directed $s-t$ network is \emph{series-parallel} if it either consists
of a single edge $(s, t)$ or can be obtained from two series-parallel
graphs with terminals $(s_1, t_1)$ and $(s_2, t_2)$ composed either in
series or in parallel. In a \emph{series composition}, $t_1$ is identified
with $s_2$, $s_1$ becomes $s$, and $t_2$ becomes $t$. In a \emph{parallel
composition}, $s_1$ is identified with $s_2$ and becomes $s$, and $t_1$ is
identified with $t_2$ and becomes $t$.
%

\section{Paradox-Free Network Topologies and Paradox-Free Nash Flows}
\label{s:easy-cases}

We start by discussing two interesting cases where Braess's paradox does not occur. We first show that if we have a bottleneck routing game $\G$ defined on an $s-t$ series-parallel network, then $\rho(\G) = 1$, and thus Braess's paradox does not occur. We recall that this was also pointed out in \cite{EFM09} for the case of atomic unsplittable bottleneck routing games.
Moreover, we note that a directed $s-t$ network is series-parallel iff it does not contain a $\theta$-graph with degree-2 terminals as a topological minor. Therefore, the example in Fig.~\ref{fig:braess} demonstrates that series-parallel networks is the largest class of network topologies for which Braess's paradox does not occur (see also \cite{Mil06} for a similar result for the case of additive routing games).
The proof of the following proposition is conceptually similar to the proof of \cite[Lemma~4.1]{EFM09}.
 
\begin{proposition}\label{prop:sepa}
Let $\G$ be bottleneck routing game on an $s-t$ series-parallel network. Then, $\rho(\G) = 1$.
\end{proposition}

\begin{proof}
Let $f$ be any Nash flow of $\G$. We use induction on the series-parallel structure of the network $G$, and show that $f$ is an optimal flow w.r.t the bottleneck cost, i.e., that $B(f) = \Bst(\G)$. For the basis, we observe that the claim holds if $G$ consists of a single edge $(s, t)$. For the inductive step, we distinguish two cases, depending on whether $G$ is obtained by the series or the parallel composition of two series-parallel networks $G_1$ and $G_2$.

\smallskip\noindent{\bf Series Composition.}
First, we consider the case where $G$ is obtained by the series composition of an $s - t'$ series-parallel network $G_1$ and a $t' - t$ series-parallel network $G_2$. We let $f_1$ and $f_2$, both of rate $r$, be the restrictions of $f$ into $G_1$ and $G_2$, respectively.

We start with the case where $B(f) = B(f_1) = B(f_2)$. Then, either $f_1$ is a Nash flow in $G_1$, or $f_2$ is a Nash flow in $G_2$. Otherwise, there would be an $s-t'$ path $p_1$ in $G_1$ with bottleneck cost $b_{p_1}(f_1) < B(f_1)$, and an $t'-t$ path $p_2$ in $G_2$, with bottleneck cost $b_{p_2}(f_2) < B(f_2)$. Combining $p_1$ and $p_2$, we obtain an $s-t$ path $p = p_1 \union p_2$ in $G$ with bottleneck cost smaller than $B(f)$, which contradicts the hypothesis that $f$ is a Nash flow of $\G$. If $f_1$ (or $f_2)$ is a Nash flow in $G_1$ (resp. $G_2$), then by induction hypothesis $f_1$ (resp. $f_2$) is an optimal flow in $G_1$ (resp. in $G_2$), and thus $f$ is an optimal flow of $\G$.

Otherwise, we assume, without loss of generality, that $B(f) = B(f_1) < B(f_2)$. Then, $f_1$ is a Nash flow in $G_1$. Otherwise, there would be an $s-t'$ path $p_1$ in $G_1$ with bottleneck cost $b_{p_1}(f_1) < B(f_1)$, which could be combined with any $t' - t$ path $p_2$ in $G_2$, with bottleneck cost $B(f_2) < B(f_1)$, into an $s-t$ path $p = p_1 \union p_2$ with bottleneck cost smaller than $B(f)$. The existence of such a path $p$ contradicts the the hypothesis that $f$ is a Nash flow of $\G$. Therefore, by induction hypothesis $f_1$ is an optimal flow in $G_1$, and thus $f$ is an optimal flow of $\G$.

\smallskip\noindent{\bf Parallel Composition.}
Next, we consider the case where $G$ is obtained by the parallel composition of an $s - t$ series-parallel network $G_1$ and an $s - t$ series-parallel network $G_2$. We let $f_1$ and $f_2$ be the restriction of $f$ into $G_1$ and $G_2$, respectively, let $r_1$ (resp. $r_2$) be the rate of $f_1$ (resp. $f_2$), and let $\G_1$ (resp. $\G_2$) be the corresponding routing instance. Then, since $f$ is a Nash flow of $\G$, $f_1$ and $f_2$ are Nash flows of $\G_1$ and $\G_2$ respectively, and $B(f_1) = B(f_2) = B(f)$. Therefore, by the induction hypothesis, $f_1$ and $f_2$ are optimal flows of $\G_1$ and $\G_2$, and $f$ is an optimal flow of $\G$. To see this, we observe that any flow different from $f$ must route more flow through either $G_1$ or $G_2$. But if the flow through e.g. $G_1$ is more than $r_1$, the bottleneck cost through $G_1$ would be at least as large as $B(f_1)$. \end{proof}\qed

Next, we show that any subpath-optimal Nash flow achieves a minimum bottleneck cost, and thus Braess's paradox does not occur if we restrict ourselves to subpath-optimal Nash flows.

\begin{proposition}\label{prop:subpathopt}
Let $\G$ be bottleneck routing game, and let $f$ be any subpath-optimal Nash flow of $\G$. Then, $B(f) = \Bst(\G)$.
\end{proposition}

\begin{proof}
Let $f$ be any subpath-optimal Nash flow of $\G$, let $S$ be the set of vertices reachable from $s$ via edges with bottleneck cost less than $B(f)$, let $\delta^+(S)$ be the set of edges $e = (u, v)$ with $u \in S$ and $v \not\in S$, and let $\delta^-(S)$ be the set of edges $e = (u, v)$, with $u \not\in S$ and $v \in S$. Then, in \cite[Lemma~4.5]{CDR06}, it is shown that (i) $(S, V \setminus S)$ is an $s-t$ cut, (ii) for all edges $e \in \delta^+(S)$, $c_e(f_e) \geq B(f)$, (iii) for all edges $e \in \delta^+(S)$ with $f_e > 0$, $c_e(f_e) = B(f)$, and (iv) for all edges $e \in \delta^-(S)$, $f_e = 0$.

By (i) and (iv), any optimal flow $o$ routes at least as much traffic as the subpath-optimal Nash flow $f$ routes through the edges in $\delta^+(S)$. Thus, there is some edge $e \in \delta^+(S)$ with $o_e \geq f_e$, which implies that $c_e(o_e) \geq c_e(f_e) \geq B(f)$, where the second inequality follows from (ii). Since $\Bst(\G) = B(o) \geq c_e(o_e)$, we obtain that $\Bst(\G) = B(f)$.
\qed\end{proof}

\section{Recognizing Paradox-Ridden Instances is Hard}
\label{s:np-hard}

In this section, we show that given a linear bottleneck routing game $\G$, it is $\NP$-hard not only to decide whether $\G$ is paradox-ridden, but also to approximate the best subnetwork within a factor less than $4/3$.
To this end, we employ a reduction from the $2$-Directed Disjoint Paths problem ($\DD$), where we are given a directed network $D$ and distinguished vertices $s_1, s_2, t_1, t_2$, and ask whether $D$ contains a pair of vertex-disjoint paths connecting $s_1$ to $t_1$ and $s_2$ to $t_2$. $\DD$ was shown $\NP$-complete in \cite[Theorem~3]{FHW80}, even if the network $D$ is known to contain two edge-disjoint paths connecting $s_1$ to $t_2$ and $s_2$ to $t_1$.
In the following, we say that a subnetwork $D'$ of $D$ is \emph{good} if $D'$ contains (i) at least one path outgoing from each of $s_1$ and $s_2$ to either $t_1$ or $t_2$, (ii) at least one path incoming to each of $t_1$ and $t_2$ from either $s_1$ or $s_2$, and (iii) either no $s_1 - t_2$ paths or no $s_2 - t_1$ paths. We say that $D'$ is \emph{bad} if any of these conditions is violated by $D'$. We note that we can efficiently check whether a subnetwork $D'$ of $D$ is good, and that a good subnetwork $D'$ serves as a certificate that $D$ is a $\YES$-instance of $\DD$. Then, the following lemma directly implies the hardness result of this section.

\begin{lemma}\label{l:gap}
Let $\I = (D, s_1, s_2, t_1, t_2)$ be any $\DD$ instance. Then, we can construct, in polynomial time, an $s-t$ network $G(V, E)$ with a linear latency function $c_e(x) = a_e x$, $a_e > 0$, on each edge $e$, so that for any traffic rate $r > 0$, the bottleneck routing game $\G = (G, c, r)$ has $\Bst(\G) = r/4$, and:
\begin{enumerate}
\item If $\I$ is a $\YES$-instance of $\DD$, there exists a subnetwork $H$ of $G$ with $B(H, r) = r/4$.

\item If $\I$ is a $\NO$-instance of $\DD$, for all subnetworks $H'$ of $G$, $B(H', r) \geq r/3$.

\item For all subnetworks $H'$ of $G$, either $H'$ contains a good subnetwork of $D$, or $B(H', r) \geq r/3$.
\end{enumerate}

\end{lemma}

\begin{proof}
We construct a network $G(V, E)$ with the desired properties by adding $4$ vertices, $s$, $t$, $v$, $u$, to $D$ and $9$ ``external'' edges $e_1 = (s, u)$, $e_2 = (u, v)$, $e_3 = (v, t)$, $e_4 = (s, v)$, $e_5 = (v, s_1)$, $e_6 = (s, s_2)$, $e_7 = (t_1, u)$, $e_8 = (u, t)$, $e_9 = (t_2, t)$ (see also Fig.~\ref{fig:redgraphgood}.a). The external edges $e_1$ and $e_3$ have latency $c_{e_1}(x) = c_{e_3}(x) = x/2$. The external edges $e_4, \ldots, e_9$ have latency $c_{e_i} = x$. The external edge $e_2$ and each edge $e$ of $D$ have latency $c_{e_2}(x) = c_{e}(x) = \eps x$, for some $\eps \in (0, 1/4)$.

\fig{0.72\textwidth}{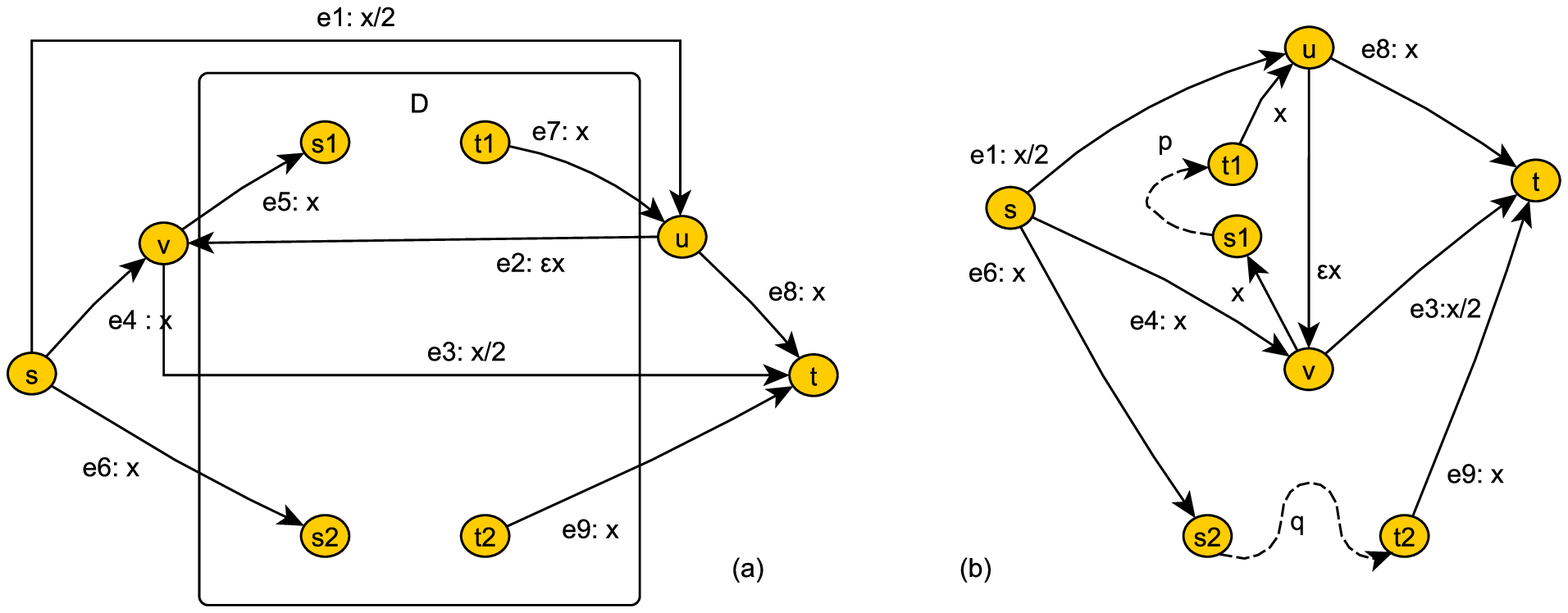}%
{\label{fig:redgraphgood}(a) The network $G$ constructed in the proof of Lemma~\ref{l:gap}. (b) The best subnetwork of $G$, with $\PoA = 1$, for the case where $D$ contains a pair of vertex-disjoint paths connecting $s_1$ to $t_1$ and $s_2$ to $t_2$.}

We first show that $\Bst(\G) = r/4$. As for the lower bound, since the edges $e_1$, $e_4$, and $e_6$ form an $s-t$ cut in $G$, every $\G$-feasible flow has a bottleneck cost of at least $r/4$. As for the upper bound, we may assume that $D$ contains an $s_1 - t_2$ path $p$ and an $s_2 - t_1$ path $q$, which are edge-disjoint (see also \cite[Theorem~3]{FHW80}).  Then, we route a flow of $r/4$ through each of the paths $(e_4, e_5, p, e_9)$ and $(e_6, q, e_7, e_8)$, and a flow of $r/2$ through the path $(e_1, e_2, e_3)$, which gives a bottleneck cost of $r/4$.

Next, we show (1), namely that if $\I$ is a $\YES$-instance of $\DD$, then 
there exists a subnetwork $H$ of $G$ 
with $B(H, r) = r/4$. By hypothesis, there is a pair of vertex-disjoint paths in $D$, $p$ and $q$, connecting $s_1$ to $t_1$, and $s_2$ to $t_2$. Let $H$ be the subnetwork of $G$ that includes all external edges and only the edges of $p$ and $q$ from $D$ (see also Fig.~\ref{fig:redgraphgood}.b). We let $\H = (H, c, r)$ be the corresponding subinstance of $\G$.
The flow routing $r/4$ units through each of the paths $(e_4, e_5, p, e_7, e_8)$ and $(e_6, q, e_9)$, and $r/2$ units through the path $(e_1, e_2, e_3)$, is an $\H$-feasible Nash flow with a bottleneck cost of $r/4$. 

We proceed to show that any Nash flow of $\H$ achieves a bottleneck cost of $r/4$. For sake of contradiction, let $f$ be a Nash flow of $\H$ with $B(f) > r/4$. Since $f$ is a Nash flow, the edges $e$ with $c_e(f_e) \geq B(f)$ form an $s-t$ cut in $H$. Since the bottleneck cost of $e_2$ and of any edge in $p$ and $q$ is at most $r/4$, this cut includes either $e_6$ or $e_9$ (or both), either $e_1$ or $e_3$ (or both), and either $e_4$ or $e_8$ (or $e_5$ or $e_6$, in certain combinations with other edges). Let us consider the case where this cut includes $e_1$, $e_4$, and $e_6$. Since the bottleneck cost of these edges is greater than $r/4$, we have more than $r/2$ units of flow through $e_1$ and more than $r/4$ units of flow through each of $e_4$ and $e_6$. Hence, we obtain that more than $r$ units of flow leave $s$, a contradiction. All other cases are similar.

To conclude the proof, we have also to show (3), namely that for any subnetwork $H'$ of $G$, if $H'$ does not contain a good subnetwork of $D$, then $B(H', r) \geq r/3$. We observe that (3) implies (2), because if $\I$ is a $\NO$-instance, any two paths, $p$ and $q$, connecting $s_1$ to $t_1$ and $s_2$ to $t_2$, have some vertex in common, and thus, $D$ includes no good subnetworks.
To show (3), we let $H'$ be any subnetwork of $G$, and let $\H'$ be the corresponding subinstance of $\G$. We first show that either $H'$ contains (i) all external edges, (ii) at least one path outgoing from each of $s_1$ and $s_2$ to either $t_1$ or $t_2$, and (iii) at least one path incoming to each of $t_1$ and $t_2$ from either $s_1$ or $s_2$, or $H'$ includes a ``small'' $s-t$ cut, and thus any $\H'$-feasible flow $f$ has $B(f) \geq r/3$.

To prove (i), we observe that if some of the edges $e_1$, $e_4$, and $e_6$ is missing from $H'$, $r$ units of flow are routed through the remaining ones, which results in a bottleneck cost of at least $r/3$. The same argument applies to the edges $e_3$, $e_8$, and $e_9$. Similarly, if $e_2$ is not present in $H'$, the edges $e_4$, $e_6$, and $e_8$ form an $s-t$ cut, and routing $r$ units of flow through them causes a bottleneck cost of at least $r/3$. Therefore, we can assume, without loss of generality, that all these external edges are present in $H'$.

Now, let us focus on the external edges $e_5$ and $e_7$. If $e_5$ is not present in $H'$ and there is a path $p$ outgoing from $s_2$ to either $t_1$ or $t_2$, routing $2r/3$ units of flow through the path $(e_1, e_2, e_3)$ and $r/3$ units through the path $(e_6, p, e_9)$ (or through the path $(e_6, p, e_7, e_8)$) is a Nash flow with a bottleneck cost of $r/3$ (see also Fig.~\ref{fig:redgraphbad}.a). If $s_2$ is connected to neither $t_1$ nor $t_2$ (no matter whether $e_5$ is present in $H'$ or not), the edges $e_1$ and $e_4$ form an $s-t$ cut, and thus, any $\H'$-feasible flow has a bottleneck cost of at least $r/3$. Similarly, we can show that if either $e_7$ is not present in $H'$, or neither $s_1$ nor $s_2$ is connected to $t_2$, any $\H'$-feasible flow has a bottleneck cost of at least $r/3$. Therefore, we can assume, without loss of generality, that all external edges are present in $H'$, and that $H'$ includes at least one path outgoing from $s_2$ to either $t_1$ or $t_2$, and at least one path incoming to $t_2$ from either $s_1$ or $s_2$.

Similarly, we can assume, without loss of generality, that $H'$ includes at least one path outgoing from $s_1$ to either $t_1$ or $t_2$, and at least one path incoming to $t_1$ from either $s_1$ or $s_2$. E.g., if $s_1$ is connected to neither $t_1$ nor $t_2$, routing $2r/3$ units of flow through the path $(e_1, e_2, e_3)$ and $r/3$ units through $s_2$ and either $t_1$ or $t_2$ (or both) is a Nash flow with a bottleneck cost of $r/3$. A similar argument applies to the case where neither $s_1$ nor $s_2$ is connected to $t_1$.

\fig{1.1\textwidth}{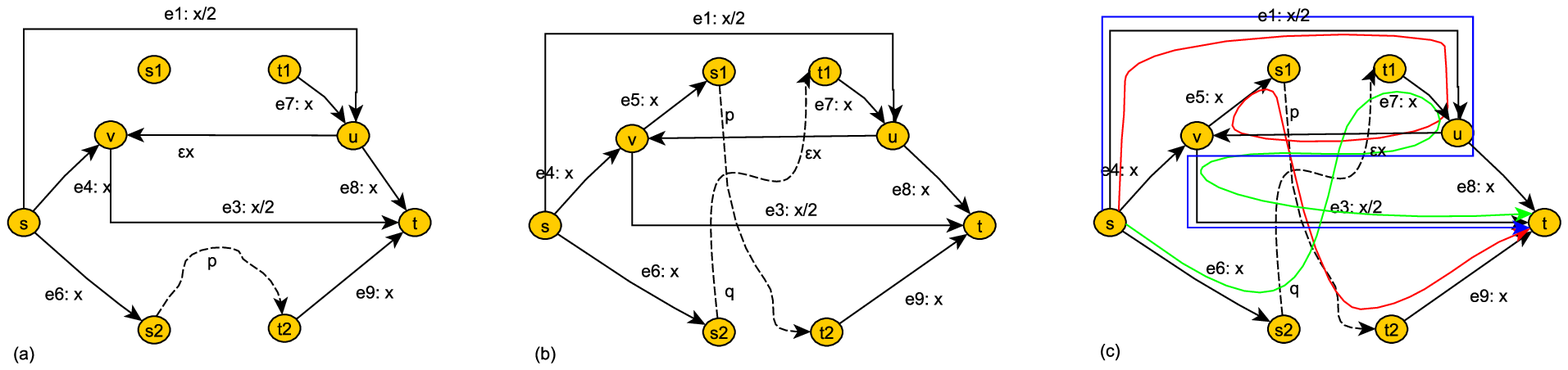}%
{\label{fig:redgraphbad}Possible subnetworks of $G$ when there is no pair of vertex-disjoint paths connecting $s_1$ to $t_1$ and $s_2$ to $t_2$. The subnetwork (a) contains an $s_2 - t_2$ path and does not include $e_5$. In the subnetwork (b), we essentially have all edges of $G$. In (c), we depict a Nash flow that consists of three paths, each carrying $r/3$ units of flow, and has a bottleneck cost of $r/3$.}

Let us now consider a subnetwork $H'$ of $G$ that does not contain a good subnetwork of $D$, but it contains (i) all external edges, (ii) at least one path outgoing from each of $s_1$ and $s_2$ to either $t_1$ or $t_2$, and (iii) at least one path incoming to each of $t_1$ and $t_2$ from either $s_1$ or $s_2$. By (ii) and (iii), and the hypothesis that the subnetwork of $D$ included in $H'$ is bad, $H'$ contains an $s_1 - t_2$ path $p$ and an $s_2 - t_1$ path $q$ (see also Fig.~\ref{fig:redgraphbad}.b).
At the intuitive level, this corresponds to the case where no edges are removed from $G$. Then, routing $r/3$ units of flow on each of the $s-t$ paths $(e_1, e_2, e_3)$, $(e_1, e_2, e_5, p, e_9)$, and $(e_6, q, e_7, e_2, e_3)$ has a bottleneck cost of $r/3$ and is a Nash flow, because the set of edges with bottleneck cost $r/3$ comprises an $s-t$ cut (see also Fig.~\ref{fig:redgraphbad}.c). Therefore, we have shown part (3) of the lemma, which in turn, immediately implies part (2). %
 \qed\end{proof}

We note that 
the bottleneck routing game $\G$ in the proof of Lemma~\ref{l:gap} has $\rho(\G) = 4/3$, and is paradox-ridden, if $\I$ is a $\YES$ instance of $\DD$, and paradox-free, otherwise. Thus, we obtain that: 

\begin{theorem}\label{th:np-hard}
Deciding whether a bottleneck routing game with strictly increasing linear latencies is paradox-ridden is $\NP$-hard.
\end{theorem}

Moreover, Lemma~\ref{l:gap} implies that it is $\NP$-hard to approximate $\Ptwo$ within a factor less than $4/3$. The subtle point here is that given a subnetwork $H$, we do not know how to efficiently compute the worst equilibrium bottleneck cost $B(H, r)$. However, we can use the notion of a good subnetwork of $D$ and deal with this issue. Specifically, let $A$ be any approximation algorithm for $\Ptwo$ with approximation ratio less than $4/3$. Then, if $D$ is a $\YES$-instance of $\DD$, $A$ applied to the network $G$, constructed in the proof of Lemma~\ref{l:gap}, returns a subnetwork $H$ with $B(H, r) < r/3$. Thus, by Lemma~\ref{l:gap}, $H$ contains a good subnetwork of $D$, which can be checked in polynomial time. If $D$ is a $\NO$-instance, $D$ contains no good subnetworks. Hence, the outcome of $A$ would allow us to distinguish between $\YES$ and $\NO$ instances of $\DD$.

\section{Approximating the Best Subnetwork is Hard}
\label{s:inapprox}

Next, we apply essentially the same construction as in the proof of Lemma~\ref{l:gap}, but in a recursive way, and show that it is $\NP$-hard to approximate $\Ptwo$ for linear bottleneck routing games within a factor of $O(n^{.121-\eps})$, for any constant $\eps > 0$.
Throughout this section, we let $\I = (D, s_1, s_2, t_1, t_2)$ be a $\DD$ instance, and let $G$ be an $s-t$ network, which includes (possibly many copies of) $D$ and can be constructed from $\I$ in polynomial time. We assume that $G$ has a linear latency function $c_e(x) = a_e x$, $a_e > 0$, on each edge $e$, and for any traffic rate $r > 0$, the bottleneck routing game $\G = (G, c, r)$ has $\Bst(\G) = r/\gamma_1$, for some $\gamma_1 > 0$. Moreover,
\begin{enumerate}
\item If $\I$ is a $\YES$-instance of $\DD$, there exists a subnetwork $H$ of $G$ with $B(H, r) = r/\gamma_1$.

\item If $\I$ is a $\NO$-instance of $\DD$, for all subnetworks $H'$ of $G$, $B(H', r) \geq r/\gamma_2$, for a $\gamma_2 \in (0, \gamma_1)$.

\item For all subnetworks $H'$ of $G$, either $H'$ contains at least one copy of a good subnetwork of $D$, or $B(H', r) \geq r/\gamma_2$.
\end{enumerate}
The existence of such a network shows that it is $\NP$-hard to approximate $\Ptwo$ within a factor less than $\gamma = \gamma_1/\gamma_2$. Thus, we usually refer to $G$ as a $\gamma$-gap instance (with linear latencies). For example, for the network $G$ in the proof of Lemma~\ref{l:gap}, $\gamma_1 = 4$ and $\gamma_2 = 3$, and thus $G$ is a $4/3$-gap instance. We next show that given $\I$ and a $\gamma_1/\gamma_2$-gap instance $G$, we can construct a $4\gamma_1/(3\gamma_2)$-gap instance $G'$, i.e., we can amplify the inapproximability gap by a factor of $4/3$.

\begin{lemma}\label{l:rec-gap}
Let $\I = (D, s_1, s_2, t_1, t_2)$ be a $\DD$ instance, and let $G$ be a $\gamma_1/\gamma_2$-gap instance with linear latencies, based on $\I$. Then, we can construct, in time polynomial in the size of $\I$ and $G$, an $s-t$ network $G'$ with a linear latency function $c_e(x) = a_e x$, $a_e > 0$, on each edge $e$, so that for any traffic rate $r > 0$, the bottleneck routing game $\G' = (G', c, r)$ has $\Bst(\G) = r/(4\gamma_1)$, and:
\begin{enumerate}
\item If $\I$ is a $\YES$-instance of $\DD$, there exists a subnetwork $H$ of $G'$ with $B(H, r) = r/(4\gamma_1)$.

\item If $\I$ is a $\NO$-instance of $\DD$, for every subnetwork $H'$ of $G'$, $B(H', r) \geq r/(3\gamma_2)$.

\item For all subnetworks $H'$ of $G'$, either $H'$ contains at least one copy of a good subnetwork of $D$, or $B(H', r) \geq r/(3\gamma_2)$.
\end{enumerate}
\end{lemma}

\begin{proof}
Starting from $D$, we obtain $G'$ by applying the construction of Lemma~\ref{l:gap}, but with all external edges, except for $e_2$, replaced by a copy of the gap-instance $G$. For convenience, we refer to the copy of the gap-instance replacing the external edge $e_i$, $i \in \{1, 3, \ldots, 9\}$, as the \emph{edgework} $G_i$.
Formally, to obtain $G'$, we start from $D$ and add four new vertices, $s$, $t$, $v$, $u$. We connect $s$ to $u$, with the $s-u$ edgework $G_1$, and $v$ to $t$, with the $s-u$ edgework $G_3$, where in both $G_1$ and $G_3$, we replace the latency function $c_e(x)$ of each edge $e$ in the gap instance with $c_e(x)/2$ (this is because in Lemma~\ref{l:gap}, the external edges $e_1$ and $e_3$ have latencies $x/2$). Moreover, instead of the external edge $e_i$, $i \in \{ 4, \ldots, 9\}$, we connect $(s, v)$, $(v, s_1)$, $(s, s_2)$, $(t_1, u)$, $(u, t)$, and $(t_2, t)$ with the edgework $G_i$. The latencies in these edgeworks are as in the gap instance. Furthermore, we add the external edge $e_2 = (u, v)$ with latency $c_{e_2}(x) = \eps x$, for some $\eps \in (0, \frac{1}{4\gamma_1})$ (see also Fig.~\ref{fig:inductionstep}.a). Also, each edge $e$ of $D$ has latency $c_{e}(x) = \eps x$. We next consider the corresponding routing instance $\G'$ with an arbitrary traffic rate $r > 0$. Throughout the proof, when we define a routing instance, we omit, for simplicity, the coordinate $c$, referring to the latency functions, with the understanding that they are defined as above.

\fig{0.85\textwidth}{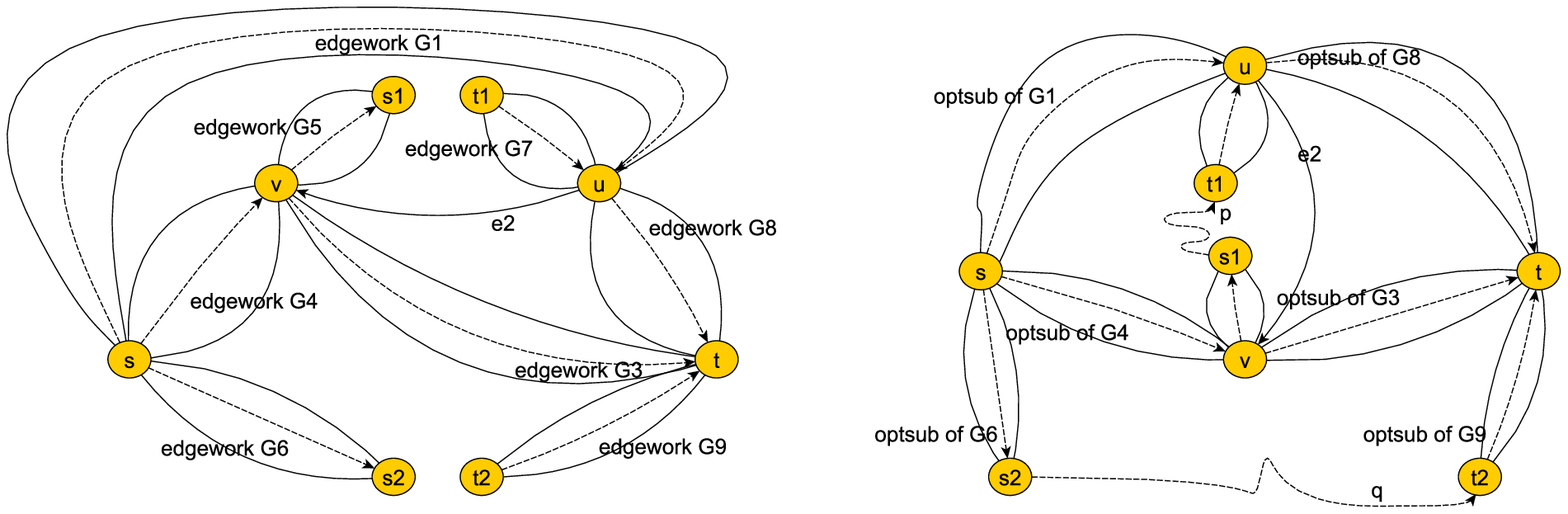}%
{\label{fig:inductionstep}(a) The network $G'$ constructed in the proof of Lemma~\ref{l:rec-gap}. The structure of $G'$ is similar to the structure of the network $G$ in Fig.~\ref{fig:redgraphgood}, with each external edge $e_i$, except for $e_2$, replaced by the edgework $G_i$. (b) The structure of a best subnetwork $H$ of $G'$, with $\PoA = 1$, when $D$ contains a pair of vertex-disjoint paths, $p$ and $q$, connecting $s_1$ to $t_1$ and $s_2$ to $t_2$. To complete $H$, we use an optimal subnetwork (or simply, subedgework) of each edgework $G_i$.}

Intuitively, each $G_i$, $i \in \{ 4, \ldots, 9 \}$, behaves as an external edge (hence the term edge(net)work), which at optimality has a bottleneck cost of $r/\gamma_1$, for any traffic rate $r$ entering $G_i$. Moreover, if $\I$ is a $\YES$-instance of $\DD$, the edgework $G_i$ has a subedgework $H_i$ for which $B(H_i, r) = r/\gamma_1$, for any $r$, while if $H_i$ does not contain any copies of a good subnetwork of $D$ (or, if $\I$ is a $\NO$-instance), for all subedgeworks $H'_i$ of $G_i$, $B(H'_i, r) \geq r/\gamma_2$, for any $r$. The same holds for $G_1$ and $G_3$, but with a worst equilibrium bottleneck cost of $r/(2\gamma_1)$ in the former case, and of $r/(2\gamma_2)$ in the latter case, because the latency functions of $G_1$ and $G_3$ are scaled by $1/2$ (see also Proposition~\ref{pr:scaling}).

The proofs of the following propositions are conceptually similar to the proofs of the corresponding claims in the proof Lemma~\ref{l:gap}.

\begin{proposition}\label{pr:gap-opt}
The optimal bottleneck cost of $\G'$ is $\Bst(\G') = r/(4\gamma_1)$.
\end{proposition}

\begin{proof}
We have to show that $\Bst(\G') = r/(4\gamma_1)$. For the upper bound, as in the proof of Lemma~\ref{l:gap}, we assume that $D$ contains an $s_1 - t_2$ path $p$ and an $s_2 - t_1$ path $q$, which are edge-disjoint. We route (i) $r/4$ units of flow through the edgeworks $G_4$, $G_5$, next through the path $p$, and next through the edgework $G_9$, (ii) $r/4$ units through the edgeworks $G_6$, next through the path $q$, and next through the edgeworks $G_7$ and $G_8$, and (ii) $r/2$ units through the edgework $G_1$, next through the external edge $e_2$, and next through the edgework $G_3$. These routes are edge(work)-disjoint, and if we route the flow optimally through each edgework, the bottleneck cost is $r/(4\gamma_1)$. As for the lower bound, we observe that the edgeworks $H_1$, $H_4$, and $H_6$ essentially form an $s-t$ cut in $G'$, and thus every feasible flow has a bottleneck cost of at least $r/(4\gamma_1)$.
\qed\end{proof}

\begin{proposition}\label{pr:gap-yes}
If $\I$ is a $\YES$-instance, there is a subnetwork $H$ of $G'$ with $B(H, r) = r/(4\gamma_1)$.
\end{proposition}

\begin{proof}
If $\I$ is a $\YES$-instance of $\DD$, then (i) there are two vertex-disjoint paths in $D$, $p$ and $q$, connecting $s_1$ to $t_1$ and $s_2$  to $t_2$, and (ii) there is an optimal subnetwork (or simply, subedgework) $H_i$ of each edgework $G_i$ so that for any traffic rate $r$ routed through $H_i$, the worst equilibrium bottleneck cost $B(H_i, r)$ is $r/\gamma_1$, if $i \in \{4, \ldots, 9\}$, and $r/(2\gamma_1)$, if $i \in \{1, 3\}$.
Let $H$ be the subnetwork of $G'$ that consists of only the edges of the paths $p$ and $q$ from $D$, of the external edge $e_2$, and of the optimal subedgeworks $H_i$, $i \in \{ 1, 3, \ldots, 9\}$ (see also Fig.~\ref{fig:inductionstep}.b). We observe that we can route: (i) $r/4$ units of flow through the subedgeworks $H_4$, $H_5$, next through the path $p$, and next through the subedgeworks $H_7$ and $H_8$, (ii) $r/4$ units of flow through the subedgework $H_6$, next through the path $q$, and next through the subedgework $H_9$, and (iii) $r/2$ units of flow through the subedgework $H_1$, next through the external edge $e_2$, and next through the subedgework $H_3$. These routes are edge(work)-disjoint, and if we use any Nash flow through each of the routing instances $(H_i, r/4)$, $i \in \{ 4, \ldots, 9\}$, $(H_1, r/2)$, and $(H_3, r/2)$, we obtain a Nash flow of the instance $(H, r)$ with a bottleneck cost of $r/(4\gamma_1)$.

We next show that any Nash flow of $(H, r)$ has a bottleneck cost of at most $r/(4\gamma_1)$. To reach a contradiction, let us assume that some feasible Nash flow $f$ has bottleneck cost $B(f) > r/(4\gamma_1)$. We recall that $f$ is a Nash flow iff the edges of $G'$ with bottleneck cost $B(f) > r/(4\gamma_1)$ form an $s-t$ cut. This cut does not include the edges of the paths $p$ and $q$ and the external edge $e_2$, due to the choice of their latencies.
Hence, this cut includes a similar cut either in $H_6$ or in $H_9$ (or in both), either in $H_1$ or $H_3$ (or in both), and either in $H_4$ or in $H_8$ (or in $H_5$ or in $H_6$, in certain combinations with other subedgeworks, see also Fig.~\ref{fig:inductionstep}.b). Let us consider the case where the edges with bottleneck cost $B(f) > r/(4\gamma_1)$ form a cut in $H_1$, $H_4$, and $H_6$. Namely, the edges of $H_1$, $H_4$, and $H_6$, with bottleneck cost equal to $B(f) > r/(4\gamma_1)$ form an $s - u$, an $s - v$, and an $s - s_2$ cut, respectively, and thus the restriction of $f$ to each of $H_1$, $H_4$, and $H_6$, is an equilibrium flow of bottleneck cost greater than $r/(4\gamma_1)$ for the corresponding routing instance. Since $\I$ is a $\YES$-instance, this can happen only if the flow through $H_1$ is more than $r/2$, and the flow through each of $H_4$ and $H_6$ is more than $r/4$ (see also property (ii) of optimal subedgeworks above). Hence, we obtain that more than $r$ units of flow leave $s$, a contradiction. All other cases are similar.
\qed\end{proof}

The most technical part of the proof is to show (3), namely that for any subnetwork $H'$ of $G'$, if $H'$ does not contain any copies of a good subnetwork of $D$, then $B(H', r) \geq r/(3\gamma_2)$. This immediately implies (2), since if $\I$ is a $\NO$-instance of $\DD$, $D$ includes no good subnetworks.
To prove (3), we consider any subnetwork $H'$ of $G'$, 
and let $H'_i$ be the subedgework of each $G_i$ present in $H'$. We assume that the subedgeworks $H'_i$ do not contain any copies of a good subnetwork of $D$, and show that if the subnetwork of $D$ connecting $s_1$ and $s_2$ to $t_1$ and $t_2$ in $H'$ is also bad, then $B(H', r) \geq r/(3\gamma_2)$.

At the technical level, we repeatedly use the idea of a flow $f_i$ through a subedgework $H'_i$ that ``saturates'' $H'_i$, in the sense that $f_i$ is a Nash flow with bottleneck cost at least $r_i/(3\gamma_2)$ for the subinstance $(H'_i, r_i)$. Formally, we say that a flow rate $r_i$ \emph{saturates} a subedgework $H'_i$ if $B(H'_i, r_i) \geq r_i / (3\gamma_2)$. We refer to the flow rate $r^s_i$ for which $B(H'_i, r^s_i) = r^s_i/(3\gamma_2)$ as the \emph{saturation rate} of $H'_i$. We note that the saturation rate $r^s_i$ is well-defined, because the latency functions of $G_i$s are linear and strictly increasing. Moreover, by property (3) of gap instances, the saturation rate of each subedgework $H'_i$ is $r^s_i \leq r/3$, if $i \in \{4, \ldots, 9\}$, and $r^s_i \leq 2r/3$, if $i \in \{1, 3\}$. Thus, at the intuitive level, the subedgeworks $H'_i$ behave as the external edges of the network constructed in the proof of Lemma~\ref{l:gap}. Hence, to show that $B(H', r) \geq r/(3\gamma_2)$, we need to construct a flow of rate (at most) $r$ that saturates a collection of subedgeworks comprising an $s-t$ cut in $H'$.

Our first step in this direction is to simplify the possible structure of $H'$.

\begin{proposition}\label{pr:gap-structure}
Let $H'$ be any subnetwork of $G'$ whose subedgeworks $H'_i$ do not contain any copies of a good subnetwork of $D$. Then, either the subnetwork $H'$ contains (i) the external edge $e_2$, (ii) at least one path outgoing from each of $s_1$ and $s_2$ to either $t_1$ or $t_2$, and (iii) at least one path incoming to each of $t_1$ and $t_2$ from either $s_1$ or $s_2$, or $B(H', r) \geq r/(3\gamma_2)$.
\end{proposition}

\begin{proof}
For convenience, in the proofs of Proposition~\ref{pr:gap-structure} and Proposition~\ref{pr:gap-routing}, we slightly abuse the terminology, and say that a collection of subedgeworks of $H'$ form an $s-t$ cut, if the union of any cuts in them comprises an $s-t$ cut in $H'$. Moreover, whenever we write that $r_i$ units of flow are routed through a subedgework $H_i$, we assume that the routing through $H_i$ corresponds to the worst Nash flow of $(H_i, r_i)$.
Also, we recall that since subedgeworks $H'_i$ do not contain any copies of a good subnetwork of $D$, by property (3) of gap instances, the saturation rate of each $H'_i$ is $r^s_i \leq r/3$, if $i \in \{4, \ldots, 9\}$, and $r^s_i \leq 2r/3$, if $i \in \{1, 3\}$.

We start by showing that either the external edge $e_2$ is present in $H'$, or $B(H', r) \geq r/(3\gamma_2)$. Indeed, if $e_2$ is not present in $H'$, the subedgeworks $H'_4$, $H'_6$, and $H'_8$ form an $s-t$ cut in $H'$. Therefore, we can construct a Nash flow $f$ that routes at least $r/3$ units of flow through $H'_4$, $H'_6$, and $H'_8$, and has $B(f) \geq r/(3\gamma_2)$. Therefore, we can assume, without loss of generality, that $e_2$ is present in $H'$.

Similarly, we show that either $H'$ includes at least one path outgoing from $s_2$ to either $t_1$ or $t_2$, and at least one path incoming to $t_2$ from either $s_1$ or $s_2$, or $B(H', r) \geq r/(3\gamma_2)$. In particular, if $s_2$ is connected to neither $t_1$ nor $t_2$, the subedgeworks $H'_1$ and $H'_4$ form an $s-t$ cut in $H'$. Thus, we can construct a Nash flow $f$ that saturates the subedgework $H'_1$ (or the subedgeworks $H'_3$ and $H'_8$, if $r^s_1 > r^s_3+r^s_8$) and the subedgework $H'_4$ (or the subedgeworks $H'_3$ and either $H'_5$, or $H'_9$ and at least one of the $H'_7$ and $H'_8$, depending on $r^s_4$ and the saturation rates of the rest). We note that this is always possible with $r$ units of flow, because $r^s_1 \leq 2r/3$ and $r^s_4 \leq r/3$. Therefore, the bottleneck cost of $f$ is $B(f) \geq r/(3\gamma_2)$. In case where there is no path incoming to $t_2$ from either $s_1$ or $s_2$, the subedgeworks $H'_3$ and $H'_8$ form an $s-t$ cut in $H'$. As before, we can construct a Nash flow $f$ that saturates the subedgeworks $H'_3$ and $H'_8$ (or, as before, an appropriate combination of other subedgeworks carrying flow to $H'_3$ and $H'_8$), and has $B(f) \geq r/(3\gamma_2)$. Therefore, we can assume, without loss of generality, that $H'$ includes at least one path outgoing from $s_2$ to either $t_1$ or $t_2$, and at least one path incoming to $t_2$ from either $s_1$ or $s_2$.

Next, we show that either $H'$ includes at least one path outgoing from $s_1$ to either $t_1$ or $t_2$, and at least one path incoming to $t_1$ from either $s_1$ or $s_2$, or $B(H', r) \geq r/(3\gamma_2)$. In particular, let us consider the case where $s_1$ is connected to neither $t_1$ nor $t_2$ (see also Fig.~\ref{fig:inductivestepbad}.a, the case where there is no path incoming to $t_1$ from either $s_1$ or $s_2$ can be handled similarly). In the following, we assume that $s_2$ is connected to $t_2$ (because, by the analysis above, we can assume that there is a path incoming to $t_2$, and $s_1$ is not connected to $T_2$), and construct a Nash flow $f$ of bottleneck cost $B(f) \geq r/(3\gamma_2)$.

We first route $\min\{ r^s_6, r^s_9 \} \leq r/3$ units of flow through the subedgework $H'_6$, next through an $s_2 - t_2$ path, and finally through the subedgework $H'_9$, and saturate either $H'_6$ or $H'_9$ (or both).
If there is an $s_2 - t_1$ path and $H'_6$ is not saturated, we keep routing flow through $H'_6$, next through an $s_2 - t_1$ path, and next through the subedgeworks $H'_7$ and $H'_8$, until either the subedgework $H'_6$ or at least one of the subedgeworks $H'_7$ and $H'_8$ become saturated.
Thus, we saturate at least one edgework on every $s-t$ path that includes $s_2$.

Next, we show how to saturate at least one edgework on every $s-t$ path that includes either $v$ or $u$. If $r^s_1 \leq r^s_3 \leq 2r/3$, we route $r^s_1$ units of flow through $H'_1$, $e_2$, and $H'_3$, and route $\min\{ r^s_3 - r^s_1, r^s_4\}$ units of flow through $H'_4$ and $H'_3$, and saturate either $H'_1$ and $H'_3$ or $H'_1$ and $H'_4$.
If $r^s_3 < r^s_1 \leq 2r/3$, we route $r^s_3$ units of flow through $H'_1$, $e_2$, and $H'_3$, and route $\min\{ r^s_3 - r^s_1, r^s_8\}$ units of flow through $H'_1$ and $H'_8$, and saturate either $H'_1$ and $H'_3$ or $H'_3$ and $H'_8$.

The remaining flow (if any) can be routed through these routes, in proportional rates. In all cases, we obtain an $s-t$ cut consisting of saturated subedgeworks. Thus, the resulting flow $f$ is a Nash flow with a bottleneck cost of at least $r/(3\gamma_2)$.
\qed\end{proof}

\fig{0.9\textwidth}{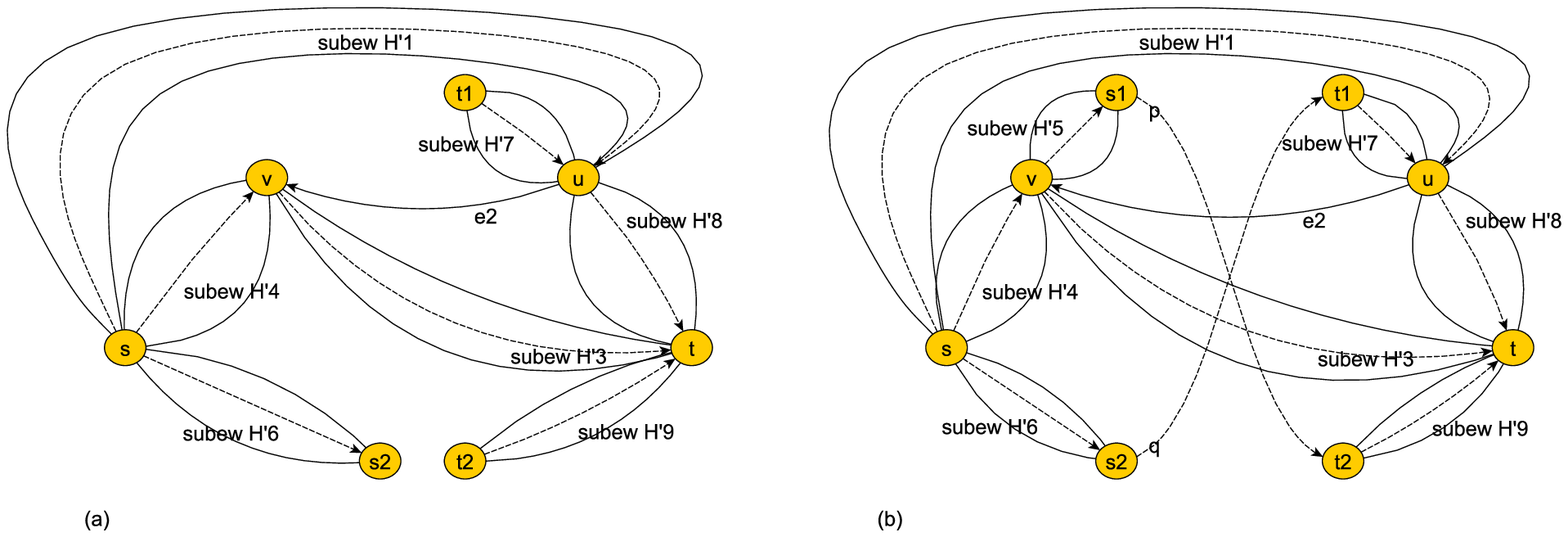}%
{\label{fig:inductivestepbad}The structure of possible subnetworks of $G'$ when there is no pair of vertex-disjoint paths connecting $s_1$ to $t_1$ and $s_2$ to $t_2$. The subnetwork (a) contains a path outgoing from $s_2$ to either $t_1$ or $t_2$, and no path outgoing from $s_1$ to either $t_1$ or $t_2$. Hence, no flow can be routed through the edgework $G_5$, and thus we can regard $G_5$ as being absent from $H'$. The subnetwork (b) essentially corresponds to the case where all edges of $G'$ are present in $H'$.}

Now, let us focus on a subnetwork $H'$ of $G'$ that contains (i) the external edge $e_2$, (ii) at least one path outgoing from each of $s_1$ and $s_2$ to either $t_1$ or $t_2$, and (iii) at least one path incoming to each of $t_1$ and $t_2$ from either $s_1$ or $s_2$. If the copy of the subnetwork of $D$ connecting $s_1$ and $s_2$ to $t_1$ and $t_2$ in $H'$ is also bad, properties (ii) and (iii) imply that $H'$ contains an $s_1 - t_2$ path $p$ and an $s_2 - t_1$ path $q$.
%
%
In this case, the entire subnetwork $H'$ essentially behaves as if it included all edges of $G'$. Then, a routing similar to that in Fig.~\ref{fig:redgraphbad}.c gives a Nash flow with a bottleneck cost of $r/(3\gamma_2)$. This intuition is formalized by the following proposition.

\begin{proposition}\label{pr:gap-routing}
Let $H'$ be any subnetwork of $G'$ that satisfies (i), (ii), and (iii) above, and does not contain any copies of a good subnetwork of $D$.  Then $B(H', r) \geq r/(3\gamma_2)$.
\end{proposition}

\begin{proof}
In the following, we consider a subnetwork $H'$ of $G'$ which does not include any copies of a good subnetwork of $D$, and contains (i) the external edge $e_2$, (ii) at least one path outgoing from each of $s_1$ and $s_2$ to either $t_1$ or $t_2$, and (iii) at least one path incoming to each of $t_1$ and $t_2$ from either $s_1$ or $s_2$. Since the copy of the subnetwork of $D$ connecting $s_1$ and $s_2$ to $t_1$ and $t_2$ in $H'$ is bad, properties (ii) and (iii) imply that $H'$ contains an $s_1 - t_2$ path $p$ and an $s_2 - t_1$ path $q$.
Moreover, since the subedgeworks $H'_i$ do not include any copies of a good subnetwork of $D$, by property (3) of gap instances, the saturation rate of each $H'_i$ is $r^s_i \leq r/3$, if $i \in \{4, \ldots, 9\}$, and $r^s_i \leq 2r/3$, if $i \in \{1, 3\}$.

We next show that for such a subnetwork $H'$, we can construct a Nash flow $f$ of bottleneck cost $B(f) \geq r/(3\gamma_2)$. At the conceptual level, as in the last case in the proof of Lemma~\ref{l:gap}, we seek to construct a Nash flow by routing $r/3$ units of flow through each of the following three routes: (i) $H'_1$, $e_2$, and $H'_3$, (ii) $H'_1$, $e_2$, $H'_5$, $p$, and $H'_9$, and (iii) $H'_6$, $q$, $H'_7$, $e_2$, and $H'_3$. However, for simplicity of the analysis, we regard the corresponding (edge) flow as being routed through just two routes: a rate of $2r/3$ is routed through $H'_1$, $e_2$, and $H'_3$, and a rate of $r/3$ is routed through the (possibly non-simple) route $H'_6$, $q$, $H'_7$, $e_2$, $H'_5$, $p$, and $H'_9$. We do so because the latter routing allows us to consider fewer cases in the analysis. We conclude the proof by showing that if the latter route is not simple, we can always decompose the flow into the three simple routes above.

In the following, we assume that with a flow rate of at most $2r/3$, routed through $H'_1$, $e_2$, and $H'_3$ (and possibly through $H'_4$ and $H'_8$), we can saturate both subedgeworks $H'_1$ and $H'_3$. Otherwise, as in the last case in the proof of Proposition~\ref{pr:gap-structure}, we can show how with a total flow rate of at most $2r/3$, part of which is routed through either $H'_4$ or $H'_8$, we can saturate either $H'_1$ and $H'_4$, or $H'_3$ and $H'_8$. Then, the remaining $r/3$ units of flow can saturate either $H'_6$, in the former case, or $H'_9$, in the latter case. Thus, we obtain a Nash flow with a bottleneck cost of at least $r/(3\gamma_2)$.

Having saturated both subedgeworks $H'_1$ and $H'_3$, using at most $2r/3$ units of flow, we have at least $r/3$ units of flow to saturate the subedgeworks $H'_5$, $H'_6$, $H'_7$, and $H'_9$, or an appropriate subset of them, so that together with $H'_1$ and $H'_3$, they form an $s-t$ cut in $H'$. We first route $\tau \equiv \min\{ r^s_5, r^s_6, r^s_7, r^s_9 \} \leq r/3$ units of flow through $H'_6$, $q$, $H'_7$, $e_2$, $H'_5$, $p$, and $H'_9$, until $t$, and consider different cases, depending on which of the subedgeworks $H'_5$, $H'_6$, $H'_7$, and $H'_9$ has the minimum saturation rate.

\begin{itemize}
\item If $\tau = r^s_9$, $H'_9$ is saturated. We first assume that $H'$ contains an $s_1 - t_1$ path, and route (some of) the remaining flow (i) through $H'_4$, $H'_5$, an $s_1-t_1$ path, $H'_7$, and $H'_8$, and (ii) through $H'_6$, $q$, $H'_7$, and $H'_8$. We do so until either at least one of the subedgeworks $H'_7$ and $H'_8$ or the subedgework $H'_6$ and at least one of the subedgeworks $H'_4$ and $H'_5$ become saturated. Since $\min\{ r^s_7, r^s_8 \} \leq r/3$, this requires at most $r/3-\tau$ additional units of flow. If $H'$ does not contain an $s_1 - t_1$ path, we route the remaining flow only through route (ii), until either at least one of the subedgeworks $H'_7$ and $H'_8$ or the subedgework $H'_6$ become saturated. In both cases, the newly saturated subedgeworks, together with the saturated subedgeworks $H'_1$, $H'_3$, and $H'_9$, form an $s-t$ cut of saturated subedgeworks, and thus the worst equilibrium bottleneck cost is at least $r/(3\gamma_2)$.

\item If $\tau = r^s_6$, $H'_6$ is saturated. As before, we first assume that $H'$ contains an $s_1 - t_1$ path, and route the remaining flow (i) through $H'_4$, $H'_5$, $p$, and $H'_9$, and (ii) through $H'_4$, $H'_5$, an $s_1-t_1$ path, $H'_9$ and $H'_8$, until either at least one of the subedgeworks $H'_4$ and $H'_5$, or the subedgework $H'_9$ and at least one of the subedgeworks $H'_7$ and $H'_8$ become saturated. Since $\min\{r^s_4, r^s_5\} \leq r/3$, this requires at most $r/3-\tau$ additional units of flow. If $H'$ does not contain an $s_1 - t_1$ path, we route the remaining flow only through route (i), until either at least one of the subedgeworks $H'_4$ and $H'_5$ or the subedgework $H'_9$ become saturated. In both cases, the newly saturated subedgeworks, together with the saturated subedgeworks $H'_1$, $H'_3$, and $H'_6$, form an $s-t$ cut of saturated subedgeworks, and thus the worst equilibrium bottleneck cost is at least $r/(3\gamma_2)$.

\item If $\tau = r^s_7$, $H'_7$ is saturated. Then, we first assume that $H'$ contains an $s_2 - t_2$ path, and route the remaining flow (i) through $H'_4$, $H'_5$, $p$, and $H'_9$, and (ii) through $H'_6$, an $s_2-t_2$ path, and $H'_9$, until either the subedgework $H'_9$, or the subedgework $H'_6$ and at least one of the subedgeworks $H'_4$ and $H'_5$ become saturated. Since $r^s_9 \leq r/3$, this requires at most $r/3-\tau$ additional units of flow. If $H'$ does not contain an $s_2 - t_2$ path, we route the remaining flow only through route (i), until either at least one of the subedgeworks $H'_4$ and $H'_5$ or the subedgework $H'_9$ become saturated. In both cases, the newly saturated subedgeworks, together with the saturated subedgeworks $H'_1$, $H'_3$, and $H'_7$, form an $s-t$ cut of saturated subedgeworks, and thus the worst equilibrium bottleneck cost is at least $r/(3\gamma_2)$.

\item If $\tau = r^s_5$, $H'_5$ is saturated. As before, we first assume that $H'$ contains an $s_2 - t_2$ path, and route the remaining flow (i) through $H'_6$, $q$, $H'_7$, and $H'_8$, and (ii) through $H'_6$, an $s_2-t_2$ path, and $H'_9$, until either the subedgework $H'_6$, or the subedgework $H'_9$ and at least one of the subedgeworks $H'_7$ and $H'_8$ become saturated. Since $r^s_6 \leq r/3$, this requires at most $r/3-\tau$ additional units of flow. If $H'$ does not contain an $s_2 - t_2$ path, we route the remaining flow only through route (i), until either at least one of the subedgeworks $H'_7$ and $H'_8$ or the subedgework $H'_6$ become saturated. In both cases, the newly saturated subedgeworks, together with the saturated subedgeworks $H'_1$, $H'_3$, and $H'_5$, form an $s-t$ cut of saturated subedgeworks, and thus the worst equilibrium bottleneck cost is at least $r/(3\gamma_2)$.
\end{itemize}

Thus, in all cases, we obtain an equilibrium flow with a bottleneck cost of at least $r/(3\gamma_2)$. However, in the construction above, the route $H'_6$, $q$, $H'_7$, $e_2$, $H'_5$, $p$, $H'_9$ may not be simple, since $p$ and $q$ may not be vertex-disjoint. If this is the case, this route is technically not allowed by our model, where the flow is only routed through simple $s-t$ paths. Nevertheless, the corresponding edge flow can be decomposed into the following three simple routes: %
(i) $H'_1$, $e_2$, and $H'_3$, (ii) $H'_1$, $e_2$, $H'_5$, $p$, and $H'_9$, and (iii) $H'_6$, $q$, $H'_7$, $e_2$, and $H'_3$, unless $\min\{ r^s_1, r^s_3 \} \leq r/3$.
Moreover, if $\min\{ r^s_1, r^s_3 \} \leq r/3$, we can work as above, and saturate both $H'_1$ and $H'_3$ with at most $r/3$ units of flow. The remaining $2r/3$ units of flow can be routed (i) through $H'_6$, $q$, $H'_7$, and $H'_8$, and (ii) through $H'_4$, $H'_5$, $p$, and $H'_9$, and possibly either through $H'_6$, an $s_2 - t_2$ path%
\footnote{We note that if the paths $p$ and $q$ are not vertex-disjoint, we also have an $s_1 -t_1$ path and an $s_2 - t_2$ path in $H'$.},
and $H'_9$, or through $H'_4$, $H'_5$, an $s_1 - t_1$ path, $H'_7$, and $H'_8$,  until either $H'_4$ (or $H'_5$) and $H'_6$, or $H'_7$ (or $H'_8$) and $H'_9$ are saturated. This routing only uses simple routes. In addition, these saturated subedgeworks, together with the saturated subedgeworks $H'_1$ and $H'_3$, form an $s-t$ cut of saturated subedgeworks, and thus the worst equilibrium bottleneck cost is at least $r/(3\gamma_2)$.
\qed\end{proof}

Propositions~\ref{pr:gap-structure}~and~\ref{pr:gap-routing} immediately imply part (3) of the lemma, which, in turn, implies part (2).
%
\qed\end{proof}

Each time we apply Lemma~\ref{l:rec-gap} to a $\gamma$-gap instance $G$, we obtain a $4\gamma/3$-gap instance $G'$ with a number of vertices of at most $8$ times the vertices of $G$ plus the number of vertices of $D$. Therefore, if we start with an instance $\I = (D, s_1, s_2, t_1, t_2)$ of $\DD$, where $D$ has $k$ vertices, and apply Lemma~\ref{l:gap} once, and subsequently apply Lemma~\ref{l:rec-gap} for $\lfloor\log_{4/3} k\rfloor$ times, we obtain a $k$-gap instance $\G'$, where the network $G'$ has $n = O(k^{8.23})$ vertices. Suppose now that there is a polynomial-time algorithm $A$ that approximates the best subnetwork of $G'$ within a factor of $O(k^{1-\eps}) = O(n^{0.121-\eps})$, for some small $\eps > 0$. Then, if $\I$ is a $\YES$-instance of $\DD$, algorithm $A$, applied to $G'$, should return a best subnetwork $H$ with at least one copy of a good subnetwork of $D$. Since $H$ contains a polynomial number of copies of subnetworks of $D$, and we can check whether a subnetwork of $D$ is good in polynomial time, we can efficiently recognize $\I$ as a $\YES$-instance of $\DD$. On the other hand, if $\I$ is a $\NO$-instance of $\DD$, $D$ includes no good subnetworks. Again, we can efficiently check that in the subnetwork returned by algorithm $A$, there are not any copies of a good subnetwork of $D$, and hence recognize $\I$ as a $\NO$-instance of $\DD$. Thus, we obtain that:

\begin{theorem}\label{th:inapprox}
For bottleneck routing games with strictly increasing linear latencies, it is $\NP$-hard to approximate $\Ptwo$ within a factor of $O(n^{0.121-\eps})$, for any constant $\eps > 0$. 
\end{theorem}

\section{Networks with Quasipolynomially Many Paths}
\label{s:sparsification}

In this section, we approximate, in quasipolynomial-time, the best subnetwork and its worst equilibrium bottleneck cost for instances $\G = (G, c, r)$ where the network $G$ has quasipolynomially many $s-t$ paths, the latency functions are continuous and satisfy a Lipschitz condition,
%
%
and the worst Nash flow in the best subnetwork routes a non-negligible amount of flow on all used edges.

We highlight that the restriction to networks with quasipolynomially many $s-t$ paths is somehow necessary, in the sense that Theorem~\ref{th:inapprox} shows that if the network has exponentially many $s-t$ paths, as it happens for the hard instances of $\DD$, and thus for the networks $G$ and $G'$ constructed in the proofs of Lemma~\ref{l:gap} and Lemma~\ref{l:rec-gap}, it is $\NP$-hard to approximate $\Ptwo$ within any reasonable factor.
%
%
In addition, we assume here that there is a constant $\delta > 0$, such that the worst Nash flow in $H^\ast$ routes more than $\delta$ units of flow on all edges of the best subnetwork $H^\ast$.

In the following, we normalize the traffic rate $r$ to $1$. This is for convenience and can be made without loss of generality%
\footnote{Given a bottleneck routing game $\G$ with traffic rate $r > 0$, we can replace each latency function $c_e(x)$ with $c_e(rx)$, and obtain a bottleneck routing game $\G'$ with traffic rate $1$, and the same Nash flows, $\PoA$, and solutions to $\Ptwo$.}.
%
%
%
Our algorithm is based on \cite[Lemma~2]{FKS09}, which applies Alth\"ofer's ``Sparsification'' Lemma \cite{Alt94}, and shows that any flow can be approximated by a ``sparse'' flow using logarithmically many paths.

\begin{lemma}\label{l:sparse-flow}
Let $\G = (G(V, E), c, 1)$ be a routing instance, and let $f$ be any $\G$-feasible flow. Then, for any $\eps > 0$, there exists a $\G$-feasible flow $\ft$ using at most $k(\eps) = \lfloor\log (2 m)/(2\eps^2)\rfloor+1$ paths, such that for all edges $e$, $|\ft_e - f_e| \leq \eps$, if $f_e > 0$, and $\ft_e = 0$, otherwise.
\end{lemma}

By Lemma~\ref{l:sparse-flow}, there exists a sparse flow $\ft$ that approximates the worst Nash flow $f$ on the best subnetwork $H^\ast$ of $G$. Moreover, the proof of \cite[Lemma~2]{FKS09} shows that the flow $\ft$ is determined by a multiset $P$ of at most $k(\eps)$ paths, selected among the paths used by $f$. Then, for every path $p \in \P$, $\ft_p = |P(p)|/|P|$, where $|P(p)|$ is number of times the path $p$ is included in the multiset $P$, and $|P|$ is the cardinality of $P$. Therefore, if the total number $|\P|$ of $s-t$ paths in $G$ is quasipolynomial, we can find, in quasipolynomial-time, by exhaustive search, a flow-subnetwork pair that approximates the optimal solution of $\Ptwo$. Based on this intuition, we next obtain an approximation algorithm for $\Ptwo$ on networks with quasipolynomially many paths, under the assumption that there is a constant $\delta > 0$, such that the worst Nash flow in the best subnetwork $H^\ast$ routes more than $\delta$ units of flow on all edges of $H^\ast$. This assumption is necessary so that the exhaustive search on the family of sparse flows of Lemma~\ref{l:sparse-flow} can generate the best subnetwork $H^\ast$, which is crucial for the analysis.


\begin{theorem}\label{th:scheme}
Let $\G = (G(V, E), c, 1)$ be a bottleneck routing game with continuous latency functions that satisfy the Lipschitz condition with a constant $\xi > 0$, let $H^\ast$ be the best subnetwork of $G$, and let $f^\ast$ be the worst Nash flow in $H^\ast$. If for all edges $e$ of $H^\ast$, $f^\ast_e > \delta$, for some constant $\delta > 0$, then for any constant $\eps > 0$, we can compute in time
\( |\P|^{O(\log(2m)/\min\{ \delta^2, \eps^2/\xi^2\})} \)
a flow $f$ and a subnetwork $H$ such that: (i) $f$ is an $\eps/2$-Nash flow in the subnetwork $H$, (ii) $B(f) \leq B(H^\ast, 1) + \eps$, (iii) $B(H, 1) \leq B(f) + \eps/4$, and (iv) $B(f) \leq B(H, 1) + \eps/2$.
\end{theorem}

\begin{proof}
Let $\eps > 0$ be a constant, and let $\e_1 = \min\{ \delta, \eps / (4\xi)\}$, and $\e_2 = \eps/2$. We show that a flow-subnetwork pair $(H, f)$ with the desired properties can be computed in time $|\P|^{O(k(\e_1))}$, where $k(\eps_1) = \lfloor \log(2m)/\min\{ 2\delta^2, \eps^2/(8\xi^2)\} \rfloor+1$,
For convenience, we say that a flow $g$ is a \emph{candidate flow} if there is a multiset $P$ of paths from $\P$, with $|P| \leq k(\e_1)$, such that $g_p = |P(p)|/|P|$, for each $p \in \P$. Namely, a candidate flow belongs to the family of sparse flows, which by Lemma~\ref{l:sparse-flow}, can approximate any other flow.
Similarly, a subnetwork $H$ is a \emph{candidate subnetwork} if there is a candidate flow $g$ such that $H$ consists of the edges used by $g$ (and only of them), and a subnetwork-flow pair $(H, g)$ is a \emph{candidate solution}, if $g$ is a candidate flow, $H$ is a candidate subnetwork that includes all the edges used by $g$ (and possibly some other edges), and $g$ is an $\e_2$-Nash flow in $H$.

By exhaustive search, in time $|\P|^{O(k(\e_1))}$, we generate all candidate flows, all candidate subnetworks, and compute the bottleneck cost $B(g)$ of any candidate flow $g$. Then, for each pair $(H, g)$, where $g$ is a candidate flow and $H$ is a candidate subnetwork, we check, in polynomial time, whether $g$ is an $\e_2$-Nash flow in $H$, and thus whether $(H, g)$ is a candidate solution.
%
%
Thus, in time $|\P|^{O(k(\e_1))}$, we determine all candidate solutions. For each candidate subnetwork $H$ that participates in at least one candidate solution, we let $\Bt(H)$ be the maximum bottleneck cost $B(g)$ of a candidate flow $g$ for which $(H, g)$ is a candidate solution. The algorithm returns the subnetwork $H$ that minimizes $\Bt(H)$, and a flow $f$ for which $(H, f)$ is a candidate solution and $\Bt(H) = B(f)$.

The exhaustive search above can be implemented in $|\P|^{O(k(\e_1))}$ time. As for the properties of the solution $(H, f)$, the definition of candidate solutions immediately implies (i), i.e., that $f$ is an $\eps/2$-Nash flow in $H$.

We proceed to show (ii), i.e., that $B(f) \leq B(H^\ast, 1) + \eps$. We recall that $H^\ast$ denotes the best subnetwork of $\G$ and $f^\ast$ denotes the worst Nash flow in $H^\ast$. Also, by hypothesis, $f^\ast_e > \delta > 0$, for all edges $e$ of $H^\ast$.

By Lemma~\ref{l:sparse-flow}, there is a candidate flow $\ft$ such that for all edges $e$ of $H^\ast$, $|\ft_e - f^\ast_e| \leq \e_1$. Thus, since $\e_1 \leq \delta$, $H^\ast$ is a candidate network, because $\ft_e > 0$ for all edges $e$ of $H^\ast$. Moreover, by the Lipschitz condition and the choice of $\e_1$, for all edges $e$ of $H^\ast$, $|c_e(\ft_e) - c_e(f^\ast_e)| \leq \eps/4$. Therefore, since $f^\ast$ is a Nash flow in $H^\ast$, $\ft$ is an $\e_2$-Nash flow in $H^\ast$, and thus $(H, \ft)$ is a candidate solution. Furthermore, $|B(\ft) - B(f^\ast)| \leq \eps/4$, i.e., the bottleneck cost of $\ft$ is within an additive term of $\eps/4$ from the worst equilibrium bottleneck cost of $H^\ast$. In particular, $B(\ft) \leq B(H^\ast, 1) + \eps/4$.

We also need to show that for any other candidate flow $g$ for which $(H^\ast, g)$ is a candidate solution, $B(g) \leq B(\ft) + 3\eps/4$, and thus $\Bt(H^\ast) \leq B(\ft) + 3\eps/4 \leq B(H^\ast, 1) + \eps$.
To reach a contradiction, let us assume that there is a candidate flow $g$ that is an $\e_2$-Nash flow in $H^\ast$ and has $B(g) > B(\ft) + 3\eps/4$. But then, we should expect that there is a Nash flow $g'$ in $H^\ast$ that closely approximates $g$ and has a bottleneck cost of $B(g') \approx B(g) > B(f^\ast)$, a contradiction.
Formally, since $g$ is an $\e_2$-Nash flow in $H^\ast$, the set of edges with $c_e(g_e) \geq B(g) - \eps/2$ comprises an $s-t$ cut in $H^\ast$. Then, by the continuity of the latency functions, we can fix a part of the flow routed essentially as in $g$, so that there is an $s-t$ cut consisting of used edges with latency $B(g) - \eps/2$, and possibly unused edges with latency at least $B(g) - \eps/2$, and reroute the remaining flow on top of it, so that we obtain a Nash flow $g'$ in $H^\ast$. But then,
\[ B(g') \geq B(g) - \eps/2 > B(\ft) + \eps/4 \geq B(f^\ast)\,, \]
which contradicts the hypothesis that $f^\ast$ is the worst Nash flow in $H^\ast$.

Therefore, $\Bt(H^\ast) \leq B(H^\ast, 1) + \eps$. Since the algorithm returns the candidate solution $(H, f)$, and not a candidate solution including $H^\ast$, $\Bt(H) \leq B(H^\ast)$. Thus, we obtain (ii), namely that $\Bt(H) = B(f) \leq B(H^\ast, 1) + \eps$.

We next show (iii), namely that $B(H, 1) \leq B(f) + \eps/4$. To this end, we let $g$ be the worst Nash flow in $H$. By Lemma~\ref{l:sparse-flow}, there is a candidate flow $\gt$ such that for all edges $e$ of $H$, $|\gt_e - g_e| \leq \e_1$, if $g_e > 0$, and $\gt_e = 0$, otherwise. Therefore, by the Lipschitz condition and the choice of $\e_1$, for all edges $e$ of $H$, $|c_e(\gt_e) - c_e(g_e)| \leq \eps/4$, if $g_e > 0$, and $c_e(\gt_e) = c_e(g_e) = 0$, otherwise. This implies that $|B(\gt) - B(g)| \leq \eps/4$, i.e., that bottleneck cost of $\gt$ is within an additive term of $\eps/4$ from the bottleneck cost of $g$. In particular, $B(g) \leq B(\gt) + \eps/4$.

We also need to show that $(H, \gt)$ is a candidate solution. Since $H$ is a candidate subnetwork and $\gt$ is a candidate flow, we only need to show that $\gt$ is an $\e_2$-Nash flow in $H$. Since $g$ is a Nash flow in $H$, the set of edges $C = \{ e: c_e(g_e) \geq B(g) \}$ comprises an $s-t$ cut in $H$. In fact, for all edges $e \in C$, $c_e(g_e) = B(g)$, if $g_e > 0$, and $c_e(g_e) \geq B(g)$, otherwise. Let us now consider the latency in $\gt$ of each edge $e \in C$. If $g_e = 0$, then $c_e(\gt_e) = c_e(g_e) \geq B(g) \geq B(\gt) - \eps/4$. If $g_e > 0$, then
\[
 B(\gt) \geq c_e(\gt_e)
        \geq c_e(g_e) - \eps/4
          =  B(g) - \eps/4
        \geq B(\gt) - \eps/2\,.
\]
Therefore, for the flow $\gt$, we have an $s-t$ cut in $H$ consisting of edges $e$ either with $\gt_e > 0$ and $B(\gt) - \eps/2 \leq c_e(\gt_e) \leq B(\gt)$, or with $\gt_e = 0$ and $c_e(\gt_e) \geq B(\gt) - \eps/4$. By the standard properties of $\eps$-Nash flows (see also in Section~\ref{s:model}), we obtain that $\gt$ is a $\e_2$-Nash flow in $H$.

Hence, we have shown that $(H, \gt)$ is a candidate solution, and that $B(g) \leq B(\gt) + \eps/4$. Therefore, the algorithm considers both candidate solutions $(H, f)$ and $(H, \gt)$, and returns $(H, f)$, which implies that $B(\gt) \leq B(f)$. Thus, we obtain (iii), namely that $B(H, 1) = B(g) \leq B(f) + \eps/4$.

To conclude the proof, we have to show (iv), namely that $B(f) \leq B(H, 1) + \eps/2$. For the proof, we use the same notation as in (iii). The argument is essentially identical to that used in the second part of the proof of (ii).
More specifically, to reach a contradiction, we assume that the candidate flow $f$, which is an $\e_2$-Nash flow in $H$, has $B(f) > B(H, 1) + \eps/2$. Then, as before, we should expect that there is a Nash flow $f'$ in $H$ that approximates $f$ and has a bottleneck cost of $B(f') \approx B(f) > B(H, 1)$, a contradiction.
Formally, since $f$ is an $\e_2$-Nash flow in $H$, the set of edges with $c_e(f_e) \geq B(f) - \eps/2$ comprises an $s-t$ cut in $H$. Then, by the continuity of the latency functions, we can fix a part of the flow routed essentially as in $f$, so that there is an $s-t$ cut consisting of used edges with latency $B(f) - \eps/2$, and possibly unused edges with latency at least $B(f) - \eps/2$, and reroute the remaining flow on top of it, so that we obtain a Nash flow $f'$ in $H$. But then,
\( B(f') \geq B(f) - \eps/2 > B(H, 1) \),
which contradicts the definition of the worst equilibrium bottleneck cost $B(H, 1)$ of $H$.
Thus, we obtain (iv), namely that $B(f) \leq B(H, 1) + \eps/2$, and conclude the proof of theorem.
\qed\end{proof}

Therefore, the algorithm of Theorem~\ref{th:scheme} returns a flow-subnetwork pair $(H, f)$ such that $f$ is an $\eps/2$-Nash flow in $H$, the worst equilibrium bottleneck cost of the subnetwork $H$ approximates the worst equilibrium bottleneck cost of $H^\ast$, since $B(H^\ast, 1) \leq B(H, 1) \leq B(H^\ast, 1)+5\eps/4$, by (ii) and (iii), and the bottleneck cost of $f$ approximates the worst equilibrium bottleneck cost of $H$, since $B(H, 1) - \eps/4 \leq B(f) \leq B(H, 1) + \eps/2$, by (iii) and (iv).

\bibliographystyle{plain}
\bibliography{braess_biblio}

\appendix
\section{Appendix}

\subsection{The Proof of Proposition~\ref{pr:scaling}}
\label{app:s:scaling}

Since the traffic rate of both $\G$ and $\G'$ is $r$, any $\G$-feasible flow $f$ is also $\G'$-feasible. Moreover, the $\G'$-latency of $f$ on each edge $e$ is $\alpha c_e(f_e)$. This immediately implies that $B_{\G'}(f) = \alpha B_{\G}(f)$, and that $f$ is a Nash flow (resp. optimal flow) of $\G$ iff $f$ is a Nash flow (resp. optimal flow) of $\G'$. \qed

\end{document}